\documentclass[12pt,draftcls,onecolumn]{IEEEtran}
\usepackage{latexsym}
\usepackage{graphicx}
\usepackage{array}
\usepackage{amsmath}
\usepackage{amsfonts}
\usepackage{amssymb}
\usepackage{amsthm}

\newtheorem{thm}{Theorem}
\newtheorem{lemma}{Lemma}
\newtheorem{prop}{Proposition}
\newtheorem{defn}{Definition}

\newtheorem{remark}{Remark}
\newtheorem{cor}{Corollary}[thm]

\newcommand{\bx} {\boldsymbol{x}}
\newcommand{\bX} {\boldsymbol{X}}

\newcommand{\bY} {\boldsymbol{Y}}
\newcommand{\bZ} {\boldsymbol{Z}}

\newcommand{\bu} {\boldsymbol{u}}

\newcommand{\sS} {\mathcal{S}}

\def\ba#1\ea{\begin{align}#1\end{align}}
\def\bal#1\eal{\begin{align}#1\end{align}}

\newcommand{{\bR}} {\right)}
\newcommand{\bp} {\begin{proof}}
\newcommand{\ep} {\end{proof}}
\newcommand{\bLF} {\left\{}
\newcommand{{\bRF}} {\right\}}

\newcommand{\uuline}[1]{\underline{\underline{#1}}}
\newcommand{\ooline}[1]{\overline{\overline{#1}}}

\begin{document}

\title{A General Formula for Compound Channel Capacity}

\author{Sergey Loyka,  Charalambos D. Charalambous

\thanks{S. Loyka is with the School of Electrical Engineering and Computer Science, University of Ottawa, Ontario, Canada, e-mail: sergey.loyka@ieee.org}

\thanks{C.D. Charalambous is with the ECE Department, University of Cyprus, Nicosia, Cyprus, e-mail: chadcha@ucy.ac.cy}

\thanks{This paper was presented in part at IEEE Int. Symp. on Information Theory (ISIT-15), Hong Kong, June 14-19, 2015, and at International Zurich Seminar
on Communications (IZS-16), March 2-4, 2016, Zurich, Switzerland.}

}

\maketitle


\begin{abstract}
A general formula for the capacity of arbitrary compound channels with the receiver channel state information is obtained using the information density approach. No assumptions of ergodicity, stationarity or information stability are made and the channel state set is arbitrary. A direct (constructive) proof is given. To prove achievability, we generalize Feinstein Lemma to the compound channel setting, and to prove converse, we generalize Verdu-Han Lemma to the same compound setting. A notion of a uniform compound channel is introduced and the general formula is shown to reduce to the familiar $\sup-\inf$ expression for such channels. As a by-product, the arbitrary varying channel capacity is established under maximum error probability and deterministic coding. Conditions are established under which the worst-case and compound channel capacities are equal so that the full channel state information at the transmitter brings in no advantage.

The compound inf-information rate plays a prominent role  in the general formula. Its properties are studied and a link between information-unstable and information-stable regimes of a compound channel is established. The results are extended to include $\varepsilon$-capacity of compound channels. Sufficient and necessary conditions for the strong converse to hold are given.
\end{abstract}

\begin{IEEEkeywords}
Channel capacity, compound channel, information stability, channel uncertainty, arbitrary-varying channel.
\end{IEEEkeywords}

\section{Introduction}

\IEEEPARstart{C}{hannel} state information (CSI) has a significant impact on channel performance  as well as code design to achieve that performance. This effect is especially pronounced for wireless channels, due to their dynamic nature, limitations of a feedback link (if any), channel estimation errors etc. \cite{Biglieri}.
When only incomplete or inaccurate CSI is available, performance analysis and coding techniques have to be modified properly. The impact of channel uncertainty  has been extensively studied since late 1950s \cite{Dobrushin}-\cite{Csiszar-92}; see \cite{Lapidoth-98A} for an extensive literature review up to late 1990s. Since channel estimation is done at the receiver (Rx) and then transmitted to the transmitter (Tx) via a limited (if any) feedback link, most studies concentrate on limited CSI available at the Tx end (CSI-T) assuming full CSI at the Rx end (CSI-R) \cite{Biglieri}, the assumption we adopt in this paper. The impact of mismatched decoding (i.e. imperfect CSI-R) on the capacity of single-state channels has been studied in \cite{Somekh-Baruch-15}.

There are several typical approaches to model channel uncertainty. In the compound channel model, the channel is unknown to the Tx but is known to belong to a certain set of channels, the uncertainty set. A member of the channel uncertainty set (state set) is selected at the beginning and held constant during the entire transmission \cite{Blackwell}-\cite{Root}, thus modeling a scenario with little dynamics (channel  coherence time significantly exceeds the codeword duration \cite{Biglieri}). A more dynamic approach is that of the arbitrary-varying channel (AVC), where the channel is allowed to vary from symbol to symbol being unknown to the Tx (but also restricted to belong to a certain class of channels) \cite{Csiszar-92}\cite{Lapidoth-98A}. A variation of the compound channel model is that of the composite channel where there is a probability assigned to each member of the compound channel set thus avoiding an over-pessimistic nature of the compound channel capacity when one channel is particularly bad but occurs with small probability \cite{Effros}. Finally, incomplete CSI at the Tx end can be addressed by assuming that the channel is not known but its distribution is known to the Tx, the so-called channel distribution information (CDI) \cite{Biglieri}.

All the studies above of compound channels require members of the uncertainty (state) set to be information-stable (e.g. stationary and ergodic), which limits significantly their applicability, especially in wireless communications, where the channel behaviour is often non-stationary, non-ergodic (as an example, many modulation-induced channels are non-stationary and quasi-static fading channels are non-ergodic). A general approach to information-unstable channels and sources (e.g. non-ergodic, non-stationary etc.), the information-spectrum approach, was pioneered in \cite{Han'93}\cite{Verdu} and developed in detail in \cite{Han}. In this paper, we apply the information-spectrum approach to extend the compound channel model \cite{Dobrushin}-\cite{Lapidoth-98A} to information-unstable scenarios, where mutual information have no operational meaning anymore. This results in a general formula for the capacity of  compound channels with arbitrary channel state sets, which are not necessarily ergodic, stationary or information-stable.

While the standard compound channel model assumes no CSI-R, such information can be obtained via a training sequence with negligible loss in rate for a quasi-static channel (which stays fixed for the entire transmission) \cite{Biglieri} provided that the uncertainty set is not too rich (without this condition, the estimation may not be possible at all, even for a quasi-static channel, as an example in Section \ref{sec.Examples} demonstrates). This justifies the compound channel model with CSI-R. On the other hand, limitations of a feedback channel (if any) result in significant uncertainty in CSI-T thus justifying the present compound channel model where no CSI is available to the Tx.

The capacity of a class of compound information-unstable channels has been studied  earlier in \cite{Han} using the information spectrum  approach. However, (i) its proof is rather involved and indirect (first, a result is established for mixed channels; then, a certain equivalence is established between mixed and compound channels, which establishes the compound channel capacity in a rather elaborate and indirect way); and (ii) its reliability criterion does not require \textit{uniform} convergence of error probability to zero (as the blocklength increases) over the whole class of channels\footnote{Uniform convergence of error probability to zero is the standard requirement for compound channels, see e.g. \cite{Blackwell}-\cite{Lapidoth-98A}\cite{Csiszar-11}, since channel state is unknown and arbitrary-low error probability is desired over the whole class of channels.}, but only for each channel individually, see Definition 3.3.1 in \cite{Han}. As a consequence, arbitrary-low error probability cannot be ensured over the whole class of (infinite-state) channels simultaneously via a sufficiently-large blocklength\footnote{In particular, when the supremum over channel states is taken, the upper bound to error probability at the bottom of p. 199 in \cite{Han} becomes infinite for infinite-state channels. Thus, Theorem 3.3.5 in \cite{Han} ensures reliable communications for finite-state channels only (see Section \ref{sec.Examples} for corresponding examples).} (in the case of finite-state channels, the convergence is automatically uniform and this problem disappears). Our approach avoids this problem by using the standard formulation of the reliability criterion for compound channels, whereby uniform convergence of error probability to zero is required over the whole class of channels simultaneously, not just for each channel individually, see Section \ref{sec.General Formula} for a detailed discussion. We obtain a general formula for the capacity of compound (possibly information-unstable) channels with arbitrary state sets (not only countable or finite) and give a direct proof by extending Feinstein and Verdu-Han Lemmas to the compound channel setting in Theorem 1 (using an algorithmic code construction).

A formulation of channel uncertainty problem based on the information density approach was presented in \cite{Effros} using the composite channel model. This, however, requires a probability measure associated with channel states, so that the channel input-output description is entirely probabilistic and the general formula in \cite{Verdu} applies to such setting. We consider the compound channel setting here, where there is no probability measure associated with channel states and a certain achievable performance has to be demonstrated for any member of the uncertainty set using a single code, for which the general formula in \cite{Verdu} is not applicable.

While the channel capacity theorem ensures the achievability of any rate below the capacity with arbitrary low error probability, there exists a hope to achieve higher rates by allowing slightly higher error probability, since the transition from arbitrary low to high error probability may be slow. Strong converse ensures that this transition is very sharp (for any rate above the capacity, the error probability converges to 1) and hence dispels the hope. In this paper, we establish the sufficient and necessary conditions for the strong converse to hold for the general compound channel. In a nutshell, the conditions require the existence of an information-stable sub-sequence of (bad) channel states (indexed by the blocklength) such that the respective sub-sequence of information densities converges in probability to the compound channel capacity. No assumptions of stationarity, ergodicity or information stability are made for the members of the uncertainty set.

The rest of the paper is organized as follows. Section \ref{sec.Channel Model} introduces a (general) channel model and assumptions. The information density approach \cite{Verdu}\cite{Han} is briefly reviewed in section \ref{sec.review}. In section \ref{sec.General Formula}, a general compound channel capacity formula is obtained in Theorem \ref{thm.C.general} using the information density approach, which holds for a wide class of channels including non-stationary, non-ergodic or information-unstable channels and arbitrary channel state sets (not only countable or finite-state). A compound inf-information rate plays a prominent role in this formula. The notion of a uniform compound channel is introduced and, for this channel, the general formula is reduced to a more familiar $\sup-\inf$ form in Theorem \ref{thm.C.uniform}. The conditions for the worst-case and compound capacities to be the same (and hence the full CSI-T to bring in no advantage) are established in Section \ref{sec.worst-case}. Section \ref{sec.Properties} presents a number of properties of the compound inf(sup)-information rate, which are instrumental to its analysis and capacity evaluation in particular scenarios. In addition to a number of inequalities, we establish the optimality of independent signalling when the compound channel is memoryless and show that the information spectrum induced by any code achieving arbitrary low error probability over the compound channel is a single atom at the code rate also equal to the mutual information rate for any channel state (so that these rates are state-independent). This links information-unstable and information-stable regimes of the compound channel.

As a by-product of the analysis, we establish the arbitrary-varying channel capacity under maximum error probability and deterministic coding with the full CSI-R, which is equal to the respective compound channel capacity (recall that the AVC capacity can be different under random and deterministic coding as well as under maximum and average error probabilities; the deterministic code AVC capacity under maximum error probability is not known in general while some special cases have been settled \cite{Lapidoth-98A}\cite{Csiszar-11}). This result shows that using average (as opposed to maximum) error probability or random (rather than deterministic) coding does not increase the AVC capacity under the full CSI-R.

In Section \ref{sec.strong.conv}, sufficient and necessary conditions for the strong converse to hold are established.  Compound $\varepsilon$-capacity is obtained in Section \ref{sec.C_eps}. The compound channel capacity is compared to that of mixed and composite channels in Section \ref{sec.Mixed.Composite} and illustrative examples are given in Section \ref{sec.Examples}. In particular, an example in Section \ref{sec.Examples}-D demonstrates that our results do not hold without the full Rx CSI assumption in general, thus demonstrating its important role.

\section{Channel Model}
\label{sec.Channel Model}

Let us consider a generic discrete-time channel model shown in Fig. 1, where $X^n = \{X_1^{(n)} ... X_n^{(n)} \}$ is a (random) sequence of $n$ input symbols, $\bX=\{X^n\}_{n=1}^{\infty}$ denotes all such sequences, and $Y^n$ is the corresponding output sequence; $s \in \sS$ denotes the channel state (which may also be a sequence) and $\sS$ is the (arbitrary) uncertainty set; $p_s(y^n|x^n)$ is the channel transition probability; $p(x^n)$ and $p_s(y^n)$ are the input and output distributions under channel state $s$.

Let us assume that the full CSI is available at the receiver but not the transmitter (see e.g. \cite{Biglieri} for a detailed motivation of this assumption; when the channel is quasi-static, i.e. stays fixed for the entire block transmission but may change for the next block, this assumption may be not necessary) and that the channel input $\bX$ and state $s$ are independent of each other. Following the standard approach (see e.g. \cite{Biglieri}), we augment the channel output with the state: $Y^n \rightarrow (Y^n,s)$. The information density \cite{Dobrushin-59}-\cite{Stratonovich-74} between the input and output for a given channel state $s$ and a given input distribution $p(x^n)$  is
\ba \notag
i(x^n;y^n,s) &= \log \frac{p(x^n,y^n,s)}{p(x^n)p(y^n,s)} \\
&= \log \frac{p_s(x^n,y^n)}{p(x^n)p_s(y^n)} \\ \notag
&= i(x^n;y^n|s)
\ea
where we have used the fact that the input $X^n$ and channel state $s$ are independent of each other. Note that we make no assumptions of stationarity, ergodicity or information stability in this paper, so that the normalized information density $n^{-1} i(X^n;Y^n|s)$ does not have to converge to the respective mutual information rate as $n \rightarrow\infty$. There is no need for the consistency assumption on $p_s(y^n|x^n)$ either (e.g. the channel may behave differently for even and odd $n$).

For future use, we give the formal definitions of information stability following \cite{Dobrushin-59}-\cite{Ting} (with a slight extension to the compound setting).

\begin{defn}
Two random sequences $\bX$ and $\bY$ are information-stable if
\ba
\frac{i(X^n;Y^n|s)}{I(X^n;Y^n|s)} \overset{\Pr}{\rightarrow} 1 \ \mbox{as}\ n \rightarrow \infty
\ea
i.e. the information density rate $\frac{1}{n} i(X^n;Y^n|s)$ converges in probability to the respective mutual information rate $\frac{1}{n} I(X^n;Y^n|s)$.
\end{defn}

\begin{defn}
Channel state $s$ is information stable if there exists an input $\bX$ such that
\ba
\frac{i(X^n;Y^n|s)}{I(X^n;Y^n|s)} \overset{\Pr}{\rightarrow} 1,\ \ \frac{I(X^n;Y^n|s)}{C_{ns}} \rightarrow 1\ \mbox{as}\ n \rightarrow \infty,
\ea
where $C_{ns} = \sup_{p(x^n)} I(X^n;Y^n|s)$ is the information capacity.
\end{defn}

As an example, a stationary discrete memoryless channel is information-stable while a non-ergodic fading channel is information-unstable in general. Information stability is both sufficient and necessary for the information capacity (and also the mutual information) to have an operational meaning \cite{Dobrushin-59}\cite{Ting} for a regular (single-state) channel.

Note that the 2nd definition requires effectively the channel to behave ergodically under the optimal input only, and tells us nothing about its behaviour under other inputs (e.g. a practical code) and, in this sense, is rather limiting. To characterize the channel behaviour under different inputs (not only the optimal one), we will consider the information stability of its input $\bX$ and the induced output $\bY$ following Definition 1 and saying that "channel is information-stable under input $\bX$". Further note that, for the compound channel, some channel states may be information stable while others are not.

We will not assume any particular noise or channel distribution so that our results are general and apply to \emph{any} such distribution.

\begin{figure}[htbp]
\centerline{\includegraphics[width=3.in]{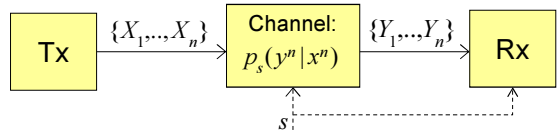}} \caption{A general discrete-time basedband  system model. No assumptions on channel state set are made. The channel is allowed to be information-unstable (e.g. non-stationary non-ergodic).}
\label{fig:Figure_channel}
\end{figure}

\section{Capacity of a Given Channel State}
\label{sec.review}

In this section, we will assume that a channel state $s$ is given and known to both the Tx and Rx (alternatively, one may assume that the channel state set is a singleton) and review the corresponding results in \cite{Verdu}\cite{Han} for this setting.

When the channel is information-stable under input $\bX$, the normalized information density converges to the mutual information rate in probability as $n\rightarrow \infty$ (due to the law of large numbers) \cite{Dobrushin-59}-\cite{Stratonovich-74},
\ba \notag
\frac{1}{n} i(X^n;Y^n|s) &\rightarrow I(\bX;\bY|s) \\
&= \lim_{n\rightarrow\infty} \frac{1}{n}\sum_{x^n,y^n} p_s(x^n,y^n) i(x^n,y^n|s)
\ea
whose operational meaning is the maximum achievable rate for a given input distribution $p(\bx)$, a channel state $s$ and arbitrary small error probability\footnote{while the summation applies to discrete alphabets, it is clear that the same argument holds for continuous alphabets using integration/probability measures instead. This applies throughout the paper unless indicated otherwise.}. Maximizing it over $p(\bx)$ results in the channel capacity. In other cases (information-unstable channels), the normalized information density remains a random variable, even when $n\rightarrow\infty$, whose support set is in general an interval \cite{Verdu}\cite{Han}. Following the analysis in \cite{Verdu}, its infimum $\underline{I}(\bX;\bY|s)$, the inf-information rate, is the largest achievable rate for a given channel state $s$, input distribution $p(\bx)$ and arbitrary-small error probability:
\ba
\underline{I}(\bX;\bY|s) \triangleq \sup_{R}\bLF R: \lim_{n\rightarrow\infty} \Pr\left\{Z_{ns} \le R\right\} =0 \bRF
\ea
where $Z_{ns}=n^{-1} i(X^n;Y^n|s)$ is the information density rate.

Following Theorems 2 and 5 in \cite{Verdu}, the channel capacity, for a given state $s$, is obtained by maximizing $\underline{I}(\bX;\bY|s)$ over $p(\bx)$,
\ba
\label{eq.C.Verdu}
C(s) = \sup_{p(\bx)} \underline{I}(\bX;\bY|s)
\ea

Note that this is a very general result, as the channel is not required to be information-stable (ergodic, stationary, etc.). The converse is proved via Verdu-Han Lemma (a lower bound to error probability, which is a dual of Feinstein bound) \cite{Verdu}\cite{Han}. We definite ($n$, $r_n$, $\varepsilon_{ns}$)-code in the standard way, where $n$ is the block length, $\varepsilon_{ns}$ is the error probability for channel state $s$ (either maximum or average error probability can be used; this has no effect on the capacity), $r_n=\ln M_n/n$ is the code rate and $M_n$ is the number of codewords.

\begin{lemma}[Verdu-Han Lemma \cite{Verdu}\cite{Han}]
\label{lemma.Verdu.Han}
Every $(n,r_n,\varepsilon_{ns})$-code  satisfies the following inequality,
\ba
\label{eq.lemma.Verdu.Han}
\varepsilon_{ns} \ge \Pr\bLF \frac{1}{n} i(X^n;Y^n|s) \le r_n - \gamma \bRF - e^{-\gamma n}
\ea
for any $\gamma > 0$, where $X^n$ is uniformly distributed over all codewords and $Y^n$ is the corresponding channel output under channel state $s$.
\end{lemma}

This is a slight re-wording of Lemma 3.2.2 in \cite{Han}, where we explicitly indicate channel state $s$ for future use.

On the other hand, the achievability of \eqref{eq.C.Verdu} for a given and known $s$ (i.e. a single, known channel) was proved in \cite{Verdu} via Feinstein Lemma.

\begin{lemma}[see e.g. \cite{Verdu}\cite{Han}]
\label{lemma.Feinstein}
For arbitrary input $X^n$, any $r_n$ and a given channel state $s$, there exists a code satisfying the following inequality,
\ba
\label{eq.lemma.Feinstein}
\varepsilon_{ns} \le \Pr\bLF \frac{1}{n} i(X^n;Y^n|s) \le r_n + \gamma \bRF + e^{-\gamma n}
\ea
for any $\gamma > 0$.
\end{lemma}

While this is sufficient to prove achievability for a given and known $s$ (codewords and decoding regions depend on channel state), it does not work for the compound channel setting, since we need a code that works for the entire class of channels, not just a single channel as in \eqref{eq.lemma.Feinstein}.

\section{Compound Channel Capacity}
\label{sec.General Formula}

In this section, we obtain a general formula for compound channel capacity of information-unstable channels by generalizing Lemmas \ref{lemma.Verdu.Han} and \ref{lemma.Feinstein} above to the compound channel setting. This will generalize the corresponding result established in \cite{Han} (Theorem 3.3.5) for finite-state channels to arbitrary compound channels. An ($n$, $r_n$, $\varepsilon_{n}$)-code for a compound channel is defined in the same way as above, with the compound error probability
\ba
\label{eq.e_n.comp}
\varepsilon_n = \sup_{s\in \sS} \varepsilon_{ns}
\ea
where $\mathcal{S}$ is the set of all possible channel states (uncertainty set), and  $\varepsilon_n \rightarrow 0$ as $n \rightarrow\infty$ is required as the reliability criterion, so that
\ba
\label{eq.lim-sup}
\lim_{n\rightarrow\infty} \sup_{s\in \sS} \varepsilon_{ns} = 0
\ea
which ensures arbitrary low error probability \textit{uniformly} over the whole class of channels for sufficiently large $n$ \cite{Biglieri}-\cite{Lapidoth-98A},
\ba
\varepsilon_{ns} \le \varepsilon\ \ \forall s \in \sS,\ \forall n \ge n_0(\varepsilon)
\ea
for any $\varepsilon>0$, where $n_0(\varepsilon)$ is a sufficiently-large blocklength. It should be emphasized that, in the compound setting, it is essential that (i) $\varepsilon_{ns} \le \varepsilon$ holds for all states $s \in \sS$ (so that the reliability is ensured uniformly over the whole class of channels) and that (ii) $n_0(\varepsilon)$ does not depend on $s$ (since the Tx does not know channel state and thus cannot choose codebooks which depend on it). On the other hand, Definition 3.3.1 in \cite{Han} does not require uniform convergence of error probability to zero over the whole class of channels so that its formulation of the reliability criterion is equivalent to
\ba
\label{eq.sup-lim}
\sup_{s\in \sS} \lim_{n\rightarrow\infty} \varepsilon_{ns} = 0
\ea
which implies $\lim_{n\rightarrow\infty} \varepsilon_{ns} = 0$ for all $s\in \sS$ and hence
\ba
\varepsilon_{ns} \le \varepsilon\ \ \forall s \in \sS,\ \forall n \ge n_0(\varepsilon,s)
\ea
i.e. $n_0(\varepsilon,s)$ depends on channel state $s$, which is in contradiction to the compound setting whereby the Tx does not know state $s$ and hence cannot use codebooks that depend on it. Hence, an arbitrary low error probability cannot be ensured simultaneously over the whole class of channels, for any blocklength, does not matter how large, under the criterion in \eqref{eq.sup-lim}. This problem disappears for finite-state channels since the convergence is automatically uniform: one can simply use $n_0(\varepsilon) = \max_s n_0(\varepsilon,s)$. Note also that \eqref{eq.sup-lim} does not imply \eqref{eq.lim-sup} in general; rather,
\ba
\lim_{n\rightarrow\infty} \sup_{s\in \sS} \varepsilon_{ns} \ge \sup_{s\in \sS} \lim_{n\rightarrow\infty} \varepsilon_{ns}
\ea
Examples of Section \ref{sec.Examples} illustrate the cases when the inequality is strict. However, \eqref{eq.sup-lim} is equivalent to \eqref{eq.lim-sup} for finite-state channels, so that Theorem 3.3.5 in \cite{Han} ensures reliable communications in that setting.

In the compound setting of this paper, \eqref{eq.lim-sup} is used as the reliability criterion, which is the standard approach \cite{Biglieri}-\cite{Lapidoth-98A}\cite{Csiszar-11},  codebooks are required to be independent of the actual channel state $s$ while the decision regions are allowed to depend on $s$ (due to the full CSI-R assumption). It is immediate that the worst-case channel capacity is $\inf_{s \in \mathcal{S}} C(s)$ but achieving this requires $s$ to be known to the Tx. If this is not the case, it is far less trivial that the compound channel capacity can be obtained by swapping $\sup$ and $\inf$; see e.g. \cite{Lapidoth-98A} for an extensive discussion of this issue. While the swapping works in many cases, there are examples when it does not \cite{Lapidoth-98B}. This is the case for the general (possibly information-unstable) compound channel considered here, whose capacity is established below.

\begin{thm}
\label{thm.C.general}
Consider the general compound channel where the channel state $s \in \mathcal{S}$ is known to the receiver but not the transmitter and is independent of the channel input; the transmitter knows the (arbitrary) uncertainty set $\mathcal{S}$. Its compound channel capacity is given by
\ba
C_c = \sup_{p(\bx)} \uuline{I}(\bX;\bY)
\ea
where $\uuline{I}(\bX;\bY)$ is the compound inf-information rate:
\ba
\label{eq.uuline{I}}
\uuline{I}(\bX;\bY) \triangleq \sup_{R}\bLF R: \lim_{n\rightarrow\infty} \sup_{s\in \sS} \Pr\left\{ Z_{ns} \le R \right\} =0 \bRF
\ea
where $Z_{ns}=n^{-1} i(X^n;Y^n|s)$ is the information density rate.
\end{thm}
\begin{proof}
To prove achievability and converse, we generalize Lemmas 1 and 2 above to the compound channel setting.

\begin{lemma}[Feinstein Lemma for compound channels]
\label{lemma.comp.Feinstein}
For arbitrary input $X^n$ and uncertainty set $\sS$ and any code rate $r_n$, there exists a $(n, r_n, \varepsilon_{n})$-code (where the codewords are independent of channel state $s$), satisfying the following inequality,
\ba
\label{eq.comp.Feinstein}
\varepsilon_{n} \le \sup_{s\in \sS} \Pr\bLF \frac{1}{n} i(X^n;Y^n|s) \le r_n + \gamma \bRF + e^{-\gamma n}
\ea
for any $\gamma > 0$.
\end{lemma}
\begin{proof} see Appendix.
\end{proof}

It is clear from the proof that the same inequality holds for both maximum and average error probability, and hence the capacity is also the same. Next, we generalize Verdu-Han Lemma to the compound channel setting.

\begin{lemma}[Verdu-Han Lemma for compound channels]
\label{lemma.comp.Verdu.Han}
For any uncertainty set $\sS$, every $(n,r_n,\varepsilon_{n})$-code  satisfies the following inequality,
\ba
\varepsilon_{n} \ge \sup_{s\in \sS} \Pr\bLF \frac{1}{n} i(X^n;Y^n|s) \le r_n - \gamma \bRF - e^{-\gamma n}
\ea
for any $\gamma > 0$, where $X^n$ is uniformly distributed over all codewords and $Y^n$ is the corresponding channel output under channel state $s$.
\end{lemma}
\begin{proof}
To prove this inequality, invoke \eqref{eq.lemma.Verdu.Han} for a given channel state $s$ and then maximize both sides over all possible channel states to obtain:
\ba
\varepsilon_{n} = \sup_s \varepsilon_{ns} \ge \sup_{s} \Pr\bLF Z_{ns} \le r_n - \gamma \bRF - e^{-\gamma n}
\ea
A subtle point here is that the original Verdu-Han Lemma allows codewords to depend on channel state while the compound codewords are independent of channel state. Since such a dependence can only decrease error probability, the desired inequality still holds.
\end{proof}

Now, to prove achievability in Theorem \ref{thm.C.general}, fix $p(\bx)$ and set $r_n \le \uuline{I}(\bX;\bY) - 2\gamma$ for any $\gamma>0$. From Lemma \ref{lemma.comp.Feinstein},
\ba
\lim_{n\rightarrow\infty} \varepsilon_{n} &\le \lim_{n\rightarrow\infty} \sup_{s\in \sS} \Pr\bLF Z_{ns} \le \uuline{I}(\bX;\bY) - \gamma \bRF = 0
\ea
which shows that $\uuline{I}(\bX;\bY) - 2\gamma$ is achievable $\forall \gamma >0$, so that $C_c \ge \sup_{p(\bx)} \uuline{I}(\bX;\bY)$.

To prove the converse, let $\uuline{I}^* = \sup_{p(\bx)}\uuline{I}(\bX;\bY)$ and select a codebook with $r_n \ge \uuline{I}^* + 2\gamma$  for some $\gamma>0$ and sufficiently large $n$, and use Lemma \ref{lemma.comp.Verdu.Han} to obtain for this codebook
\ba \notag
\lim_{n\rightarrow\infty}\varepsilon_{n} &\ge \lim_{n\rightarrow\infty}\sup_{s\in \sS} \Pr\bLF Z_{ns} \le \uuline{I}^* + \gamma \bRF\\ \notag
 &\ge \lim_{n\rightarrow\infty}\sup_{s\in \sS} \Pr\bLF Z_{ns} \le \uuline{I}(\bX;\bY) + \gamma \bRF\\
&\ge \varepsilon_0 > 0
\ea
for some fixed $\varepsilon_0 > 0$, where the last two inequalities follow from the definition of $\uuline{I}$ and 2nd inequality follows from $\uuline{I}^* \ge \uuline{I}(\bX;\bY)$, so that no rate above $\uuline{I}^*$ is achievable: $C_c \le \uuline{I}^*$.

It is clear from the proof that the same capacity holds under the maximum as well as average error probability.
\end{proof}

\begin{remark}
It is $\uuline I(\bX,\bY)$ that extends $\underline{I}(\bX,\bY|s)$ to the compound channel setting, not $\underline{I}(\bX,\bY) \triangleq  \inf_s \underline{I}(\bX,\bY|s)$, in the general case.
\end{remark}

The relationship between $\uuline I(\bX,\bY)$ and $\underline{I}(\bX,\bY)$ is established below.

\begin{prop}
\label{prop.uuline(I)<=underline(I)}
The following inequality holds for a general compound channel
\ba
\label{eq.prop.uuline(I)<=underline(I)}
\uuline I(\bX,\bY) \le \underline{I}(\bX,\bY) \triangleq \inf_s \underline{I}(\bX,\bY|s)
\ea
\end{prop}
\begin{proof}
The proof is by contradiction. Let $\uuline I = \uuline I(\bX,\bY)$, $\underline{I} = \underline{I}(\bX,\bY)$ and assume that $\uuline I > \underline{I}$, set $R=(\uuline I + \underline{I})/2 > \underline{I}$ and observe that $R< \uuline{I}$ and therefore
\ba\notag
\lim_{n\rightarrow\infty}\sup_{s} \Pr\bLF Z_{ns} \le R \bRF &\ge \sup_{s} \lim_{n\rightarrow\infty}\Pr\bLF Z_{ns} \le R \bRF\\
&\ge \varepsilon_0 > 0
\ea
for some $\varepsilon_0 > 0$ - a contradiction, where the last two inequalities are from the definition of $\underline{I}$. Therefore, $\uuline I \le \underline{I}$.
\end{proof}

%
\subsection{Uniform compound channels}
It can be demonstrated, via examples (see Examples 1 and 2 in Section \ref{sec.Examples}), that the inequality in \eqref{eq.prop.uuline(I)<=underline(I)} can be strict. To see when the equality is achieved, we need the following definition.

\begin{defn}
\label{defn.uniform.ch}
A compound channel is uniform if there exists $\delta\ge 0$ such that for any input $\bX_{\delta}$ satisfying $\uuline{I}(\bX_\delta;\bY_{\delta}) \ge C_c- \delta$ (i.e.  $\bX_{\delta}$ is $\delta$-suboptimal), where $\bY_{\delta}$ is the corresponding output, the convergence in
\ba
\label{eq.uniform.conf}
\Pr\bLF n^{-1} i(X^n_{\delta};Y^n_{\delta}|s) \le  \underline{I}(\bX_{\delta},\bY_{\delta}) - \gamma \bRF \rightarrow 0
\ea
as $n\rightarrow\infty$ is uniform in $s\in \sS$ for all sufficiently small $\gamma >0$.
\end{defn}

Note that while the point-wise convergence is ensured for each $s$ from the definition of $\underline{I}(\bX_{\delta},\bY_{\delta})$, it does not have to be uniform and, indeed, examples can be constructed where it is not (see Section \ref{sec.Examples}). In a sense, the uniform convergence here ensures that the channel does not behave "too badly" as $n$ increases. It is straightforward to see that if the uniform convergence in \eqref{eq.uniform.conf} holds for some $\gamma=\gamma_0 >0$, then it also holds for any $\gamma >\gamma_0$, so that the condition needs to be checked for arbitrary small $\gamma>0$ only. If the supremum in $C_c= \sup_{p(\bx)} \uuline{I}(\bX,\bY)$ is achieved, then one may take $\delta=0$ and use the optimal input only. All finite-state compound channels are uniform under any input (i.e. one may take $\delta=C_c$).

For a uniform compound channel, one obtains the following result.

\begin{prop}
\label{prop.uuline(I)=underline(I)}
The following equality holds for any $\bX_{\delta}$ if and only if the compound channel is uniform,
\ba
\label{eq.prop.uuline(I)=underline(I)}
\uuline I(\bX_{\delta},\bY_{\delta}) = \underline{I}(\bX_{\delta},\bY_{\delta})
\ea
If $\delta = C_c$, then this holds for any input.
\end{prop}
\begin{proof}
see Appendix.
\end{proof}

We are now in a position to establish the capacity of uniform compound channels.

\begin{thm}
\label{thm.C.uniform}
Consider the general compound channel where the channel state $s \in \mathcal{S}$ is known to the receiver but not the transmitter and is independent of the channel input; the transmitter knows the (arbitrary) uncertainty set $\mathcal{S}$. Its compound channel capacity is bounded by
\ba
\label{eq.thm.C.uniform}
C_c \le \sup_{p(\bx)} \inf_{s\in \sS} \underline{I}(\bX;\bY|s)
\ea
with equality for a uniform compound channel. In particular, this holds when $\sS$ is of finite cardinality.
\end{thm}
\begin{proof}
The inequality follows from \eqref{eq.prop.uuline(I)<=underline(I)}. The equality part is established by using Proposition \ref{prop.uuline(I)=underline(I)} in Theorem \ref{thm.C.general} (note that taking the supremum over all $\bX_{\delta}$ is sufficient). It is straightforward to verify that a finite cardinality of $\sS$ implies the uniform convergence condition in \eqref{eq.uniform.conf} for any input (not only $\delta$-suboptimal).
\end{proof}

As far as the compound channel capacity is concerned, the uniform convergence condition in \eqref{eq.uniform.conf} needs to hold for optimal or suboptimal inputs only for \eqref{eq.thm.C.uniform} to hold with equality. Note also that Theorems \ref{thm.C.general} and \ref{thm.C.uniform} hold for any alphabet and any uncertainty set. In many cases of practical interest (e.g. when the set of feasible input distributions $p(\bx)$ and/or the uncertainty set $\mathcal{S}$ are compact and $\underline{I}(\bX;\bY|s)$ is well-behaving), $\sup$ and/or $\inf$ can be substituted by $\max$ and/or $\min$. Unlike Theorem 3.3.5 in \cite{Han}, the present result applies to arbitrary channel uncertainty sets and its proof is direct (i.e. not relying on mixed channels but directly constructing capacity-approaching codes for compound channels in Lemma \ref{lemma.comp.Feinstein}). The examples in Section \ref{sec.Examples} demonstrate that the inequality can be strict.

We remark that many well-known results (e.g. \cite{Root}) are  special cases of Theorem \ref{thm.C.general} and \ref{thm.C.uniform}. The latter is pleasantly similar to known results for information-stable channels, which also include $\sup - \inf$ expression. When $\sS$ is of finite cardinality, \eqref{eq.thm.C.uniform} coincides with the compound capacity in Theorem 3.3.5 in \cite{Han}, i.e. the compound and mixed channels have the same capacity in this case. Examples 1 and 2 in Section \ref{sec.Examples} show that the compound capacity can be strictly less than the corresponding mixed channel capacity in the general case.

One may ask whether the $\sup-\inf$ capacity formula in Theorems \ref{thm.C.uniform} apply to a broader class of channels than those in Definitions \ref{defn.uniform.ch}, i.e. without imposing the uniform convergence condition. We consider this below.

\begin{defn}
A sequence of functions $f_n(s)$ is weakly decreasing if there exists $\delta_m \ge 0$ such that $\delta_m \rightarrow 0$ as $m \rightarrow \infty$ and
\bal
f_n(s) \le f_m(s) +\delta_m\ \forall n\ge m,\ \forall s
\eal
\end{defn}

\begin{prop}
If the uncertainty set $\mathcal{S}$ is compact (e.g. closed and bounded) and there exists such $\delta \ge 0$ that
\bal
f_n(s) = \Pr\bLF n^{-1} i(X^n_{\delta};Y^n_{\delta}|s) \le \underline{I}(\bX_{\delta};\bY_{\delta}) - \gamma \bRF,
\eal
is upper semi-continuous in $s$ and weakly decreasing for all sufficiently small $\gamma>0$ and all sufficiently large $n$, and for any $\delta$-suboptimal input $\bX_{\delta}$, then \eqref{eq.prop.uuline(I)=underline(I)} holds for any $\bX_{\delta}$ and hence the equality in \eqref{eq.thm.C.uniform} follows.
\end{prop}
\begin{proof}
Using Theorem A.1.5(b) in \cite{Bauerle} under the stated conditions ensures the 1st equality in \eqref{eq.lim sup = sup lim} while the 2nd equality follows from the definition of $\underline{I}(\bX_{\delta};\bY_{\delta})$, from which the first statement follows. The 2nd statement can be obtained by observing that the supremum can be taken over $\bX_{\delta}$ only without any loss.
\end{proof}

It is straightforward to see that the uniform convergence in Definition \ref{defn.uniform.ch} implies the weakly-decreasing property but the converse is not necessarily true. On the other hand, there is no requirement for $\sS$ to be compact in Definition \ref{defn.uniform.ch}, so that these formulations are complementary to each other. It can be shown that any finite-state compound channel is uniform and thus a special case for Theorems \ref{thm.C.uniform} and \ref{thm.Ce.uniform}. The weakly-decreasing property represents the natural case where the performance improves  with blocklength while the continuity property holds for many channel models. Note that $\mathcal{S}$ is not required here to be countably-finite or even countable (but it has to be bounded and closed).

\subsection{Worst-case channel capacity}
\label{sec.worst-case}

One may also consider the worst-case channel capacity $C_w$ (i.e. the capacity of the worst-case channel in the uncertainty set),
\bal
\label{eq.Cw}
C_w = \inf_{s\in \sS} \sup_{p(\bx)} \underline{I}(\bX;\bY|s)
\eal
which has the operational meaning under the full Tx CSI. It is well-known that $C_w \ge C_c$ (since any code for the compound channel must also work on the worst-case channel) and there are many cases where the inequality is strict. Below, we establish conditions under which they are equal for the general compound channel.

\begin{defn}
A saddle-point property is said to hold if
\bal
\inf_{s\in \sS} \sup_{p(\bx)} \underline{I}(\bX;\bY|s) = \sup_{p(\bx)} \inf_{s\in \sS} \underline{I}(\bX;\bY|s)
\eal
\end{defn}

Note that this definition does not impose any operational meaning on the quantities involved. The following proposition establishes the conditions under which $C_w = C_c$ for the general compound channel.

\begin{prop}
Consider the general compound channel under the full Rx CSI such that: (i) the saddle-point property holds, and (ii) the compound channel is uniform. Then, the worst-case and compound capacities are the same,
\bal
\label{eq.Cw=Cc}
C_w = \inf_{s\in \sS} \sup_{p(\bx)} \underline{I}(\bX;\bY|s) = \sup_{p(\bx)} \uuline{I}(\bX;\bY) = C_c
\eal
The 1st condition is also necessary.
\end{prop}
\begin{proof}
Consider the following chain inequality:
\bal\notag
C_w &= \inf_{s\in \sS} \sup_{p(\bx)} \underline{I}(\bX;\bY|s)\\ \notag
 &\ge \sup_{p(\bx)}\inf_{s\in \sS} \underline{I}(\bX;\bY|s)\\
 &\ge \sup_{p(\bx)} \uuline{I}(\bX;\bY) = C_c
\eal
where the 2nd inequality is due to \eqref{eq.prop.uuline(I)<=underline(I)}, and observe that the inequalities become the equalities under the conditions in (i) and (ii).
\end{proof}

The significance of this result is due to the fact that while achieving the worst-case capacity allows the codebooks to depend on the channel state, achieving the compound channel capacity does not allow this, so that the presence of the full Tx CSI does not bring in any advantage in this case. It can be further extended as follows.

\begin{defn}
A compound channel is (stochastically) degraded if there exists such channel state $s_w$ that is degraded with respect to any other channel state $s$ in the uncertainty set, i.e. if there exists such fictitious channel $q_s(y^n_{s_w}|y^n_s)$ that
\bal
p_{s_w}(y^n_{s_w}|x^n) = \sum_{y^n_s} p_s(y^n_s|x^n) q_s(y^n_{s_w}|y^n_s)
\eal
e.g. if $X^n \rightarrow Y^n_s \rightarrow Y^n_{s_w}$ is a Markov chain for any $s$ and any $n$.
\end{defn}

\begin{prop}
If the general compound channel is degraded, then its worst-case and compound capacities are same, as in \eqref{eq.Cw=Cc}.
\end{prop}
\begin{proof}
In general, $C_w \ge C_c$. For a degraded compound channel, any code that is good for the worst-case channel, is also good for any other channel in the uncertainty set (since the receiver can emulate the artificial channel $q_s(y^n_{s_w}|y^n_s)$ while making the decisions) and hence $C_w \le C_c$, from which the equality follows.
\end{proof}

\section{Properties of Compound Inf-Information Rate}
\label{sec.Properties}

Below we study the properties of the compound inf-information rate $\uuline{I}(\bX,\bY)$, which are instrumental in evaluating this quantity and the compound channel capacity for specific channels.

First, let us establish inequalities for compound random sequences (i.e. sequences of random variables indexed by a common state) which are instrumental for further development. We will need the following definitions.

\begin{defn}
Let $\bX=\{X_{sn}\}_{n=1}^{\infty}$ be a compound random sequence where $s$ is a state. The compound infimum $\uuline{\{\cdot\}}$ and supremum $\ooline{\{\cdot\}}$ operators are defined as follows:
\ba
\label{eq.comp.inf-sup}
\uuline{\bX}=\uuline{\{X_{sn}\}} &\triangleq \sup \bLF x: \lim_{n\rightarrow\infty} \sup_{s} \Pr\left\{ X_{sn} \le x \right\} =0 \bRF\\
\ooline{\bX}=\ooline{\{X_{sn}\}} &\triangleq \inf \bLF x: \lim_{n\rightarrow\infty} \sup_{s} \Pr\left\{ X_{sn} \ge x \right\} =0 \bRF
\ea
\end{defn}
These operators generalize the respective sup $\overline{\bX}$ and inf $\underline{\bX}$ operators for regular (single-state) sequences. They have the following important properties, which facilitate their evaluation and analysis.

\begin{prop}
\label{prop.prop.XY}
Let $\{X_{ns}\}_{n=1}^{\infty}$ and $\{Y_{ns}\}_{n=1}^{\infty}$ be two (arbitrary) compound random sequences and $s$ is a (common) state. Then, the following holds:
\ba
\label{eq.prop.1}
&\uuline{\bX}\le \ooline{\bX},\\
\label{eq.prop.2}
&\ooline{\bX} = -\uuline{(-\bX)},\\
\label{eq.prop.3} \notag
&\uuline{\bX}+\uuline{\bY} \le \uuline{(\bX+\bY)}\\ \notag
&\qquad\quad\ \le \min\{\uuline{\bX} + \ooline{\bY}, \ooline{\bX} + \uuline{\bY}\}\\ \notag
&\qquad\quad\ \le \uuline{\bX} + \ooline{\bY}\\
&\qquad\quad\ \le \ooline{\bX} + \ooline{\bY},\\
\label{eq.prop.4}\notag
&\ooline{\bX}+\ooline{\bY} \ge \ooline{(\bX+\bY)}\\  \notag
 &\qquad\quad\ \ge \max\{\uuline{\bX} + \ooline{\bY}, \ooline{\bX} + \uuline{\bY}\}\\ \notag
  & \qquad\quad\ \ge \ooline{\bX}+\uuline{\bY}\\
& \qquad\quad\ \ge \uuline{\bX}+\uuline{\bY}
\ea
\end{prop}
\begin{proof}
See Appendix.
\end{proof}

\begin{remark}
Note that the inequalities in Proposition \ref{prop.prop.XY} do not follow directly from the respective inequalities for $\underline{(\bX+\bY)}$ in \cite{Han} for single-state sequences since (i) $\sup_s$ may result in different maximizing states for $X_{ns}, Y_{ns}$ and $X_{ns}+Y_{ns}$ sequences, and (ii) $\lim$ and $\sup$ may not be swapped in general (unless the uniform convergence holds, in which case the compound inequalities can be obtained from non-compound ones in \cite{Han} by using an equality similar to that in \eqref{eq.prop.uuline(I)=underline(I)}).
\end{remark}

The following result will be needed below.

\begin{prop}
\label{prop.zns}
Consider a compound random sequence $\{Z_{ns}\}_{n=1}^{\infty}$ where $\sigma_{ns}^2$ is the variance of $Z_{ns}$ such that
\ba
\label{eq.sigmans=0}
\lim_{n\rightarrow\infty}\sup_ s \sigma_{ns}^2 =0
\ea
Then,
\ba
\uuline{\bZ} \triangleq \uuline{\{Z_{ns}\}} = \liminf_{n\rightarrow\infty}\inf_ s E\{Z_{ns}\} \triangleq \tilde{Z}
\ea
\end{prop}
\begin{proof}
See Appendix.
\end{proof}

Note that Proposition \ref{prop.zns} equates two very different quantities: one includes no averaging ($\uuline{\bZ}$) and the other is based on averaging ($\tilde{Z}$).

To proceed further, we extend the definitions in \cite{Verdu}\cite{Han} to the compound setting here.

\begin{defn}
Let $X^n$ and $Y^n$ be two compound random sequences with distributions $p_{s x^n}$ and $p_{s y^n}$ where $s$ is a state. The compound inf-divergence rate is defined as
\ba
\label{eq.D(x||y)}
\uuline{D}(\bX;\bY) \triangleq \uuline{\bLF \frac{1}{n} \ln \frac{p_{s x^n}(X^n)}{p_{s y^n}(X^n)}\bRF }
\ea
and likewise for the compound inf-entropy rate $\uuline{H}(\bX)$ and sup-entropy rate $\ooline{H}(\bX)$:
\ba
\uuline{H}(\bX) \triangleq \uuline{\{h_{sn}(X^n)\}}, \ \ \ooline{H}(\bX) \triangleq \ooline{\{h_{sn}(X^n)\}},
\ea
where $h_{sn}(x^n) = -n^{-1} \ln p_{s x^n}(x^n)$. The compound conditional inf-entropy rate $\uuline{H}(\bY|\bX)$ and sup-entropy rate $\ooline{H}(\bY|\bX)$ are defined analogously (with respect to joint distribution $p_{s x^n y^n}$), and $\ooline{I}(\bX;\bY)$ is similarly defined.
\end{defn}

The proposition below gives the properties useful in evaluation of compound inf-information rate $\uuline{I}(\bX;\bY)$ (which extend the respective properties in \cite{Verdu} to the compound setting).

\begin{prop}
\label{prop.properties.I}
Let $\bX$, $\bY$ and $\bZ$ be (arbitrary) compound random sequences. The following holds:
\ba
\label{eq.properties.1}
&\uuline{D}(\bX||\bY) \ge 0 \\
\label{eq.properties.2}
&\ooline{I}(\bX;\bY) \ge \uuline{I}(\bX;\bY) \ge 0 \\
\label{eq.properties.3}
&\uuline{I}(\bX;\bY) = \uuline{I}(\bY;\bX) \\
\label{eq.properties.4}
&\uuline{I}(\bX;\bY) \le \ooline{H}(\bY) - \ooline{H}(\bY|\bX)\\
\label{eq.properties.5}
&\uuline{I}(\bX;\bY) \le \uuline{H}(\bY) - \uuline{H}(\bY|\bX)\\
\label{eq.properties.6}
&\uuline{I}(\bX;\bY) \ge \uuline{H}(\bY) - \ooline{H}(\bY|\bX)\\
\label{eq.properties.4a}
&\ooline{H}(\bY) \ge \ooline{H}(\bY|\bX)\\
\label{eq.properties.5a}
&\ooline{H}(\bY) \ge \uuline{H}(\bY) \ge \uuline{H}(\bY|\bX)\\
\label{eq.properties.7}
&\uuline{I}(\bX,\bY;\bZ) \ge \uuline{I}(\bX;\bZ) + \uuline{I}(\bY;\bZ|\bX) \ge \uuline{I}(\bX;\bZ)
\ea
with equality if $\ooline{I}(\bY;\bZ|\bX)=0$.

If the alphabets are discrete, then
\ba
\label{eq.properties.8}
&0\le \uuline{H}(\bX|\bY) \le \uuline{H}(\bX) \le \ooline{H}(\bX) \le \ln N_x\\
\label{eq.properties.9}\notag
&0\le \uuline{I}(\bX;\bY) \le \min\{\uuline{H}(\bX), \uuline{H}(\bY)\}\\
 &\qquad\qquad\qquad \le \min\{\ln N_x, \ln N_y\} \\
\label{eq.properties.9a}\notag
&\uuline{I}(\bX;\bY) = \min\{\uuline{H}(\bX), \uuline{H}(\bY)\}\\
 &\qquad\qquad\qquad\ \mbox{if}\ \min\{\ooline{H}(\bY|\bX), \ooline{H}(\bX|\bY)\}=0 \\
\label{eq.properties.10}\notag
&0\le \ooline{I}(\bX;\bY) \le \min\{\ooline{H}(\bX), \ooline{H}(\bY)\}\\
&\qquad\qquad\qquad\le \min\{\ln N_x, \ln N_y\}
\ea
where the last inequalities in \eqref{eq.properties.8}-\eqref{eq.properties.10} hold if the alphabets are of finite cardinality $N_x, N_y$.
\end{prop}
\begin{proof}
See Appendix.
\end{proof}

Note that many of these properties mimic the respective properties of mutual information and entropy, e.g. "conditioning cannot increase the entropy" and "mutual information is non-negative, symmetric and bounded by the entropy of the alphabet". Similar properties can also be established for compound sup-information rate $\ooline{I}(\bX;\bY)$. The next Proposition establishes the data processing inequality in terms of compound inf-information rates.

\begin{prop}[Data processing inequality]
Let $\bX \rightarrow \bY \rightarrow \bZ$ be a compound Markov chain. Then,
\ba
\uuline{I}(\bX;\bY) \ge \uuline{I}(\bX;\bZ)
\ea
with equality if $\ooline{I}(\bX;\bY|\bZ)=0$.
\end{prop}
\begin{proof}
Observe that
\ba\notag
i(x^n;y^n,z^n|s) &= \ln \frac{p_s(x^n|y^n z^n)}{p_s(x^n)}\\
 &= \ln \frac{p_s(x^n|y^n)}{p_s(x^n)}\\ \notag
& = i(x^n;y^n|s)
\ea
where 2nd equality is due to conditional independence of $X^n$ and $Z^n$ given $Y^n$, and that
\ba\notag
i(x^n;y^n,z^n|s) &= \ln \frac{p_s(x^n|z^n)}{p_s(x^n)} + \ln \frac{p_s(x^n|y^n z^n)}{p_s(x^n|z^n)}\\
&= i(x^n;z^n|s)+ i(x^n;y^n|z^n s)
\ea
so that
\ba
i(x^n;y^n|s) = i(x^n;z^n|s)+ i(x^n;y^n|z^n s)
\ea
Taking $\uuline{(\cdot)}$ of both sides and using the inequality in \eqref{eq.prop.3}, one obtains
\ba
\label{eq.prop.data.process.5}
\uuline{I}(\bX,\bY) \ge \uuline{I}(\bX;\bZ) + \uuline{I}(\bX;\bY|\bZ) \ge \uuline{I}(\bX;\bZ)
\ea
where the last inequality is due to $\uuline{I}(\bX;\bY|\bZ) \ge 0$. To prove the equality part, observe that
\ba
\uuline{I}(\bX,\bY) \le \uuline{I}(\bX;\bZ) + \ooline{I}(\bX;\bY|\bZ) = \uuline{I}(\bX;\bZ)
\ea
and use \eqref{eq.prop.data.process.5}.
\end{proof}

Next Proposition links the compound inf-information rate to the mutual information rates.

\begin{prop}
\label{prop.uulineI<=liminf.inf.I}
Consider the general compound channel. Its compound inf-information rate is bounded as follows:
\ba
\label{eq.uulineI<=liminf.inf.I}\notag
\uuline{I}(\bX,\bY) &\overset{(a)}{\le} \liminf_{n\rightarrow\infty} \inf_s \frac{1}{n}I(X^n;Y^n|s)\\
 &\overset{(b)}{\le} \liminf_{n\rightarrow\infty} \inf_s \frac{1}{n}\sum_{k=1}^n I(X_k;Y_k|s)\\ \notag
&\overset{(c)}{\le} \liminf_{n\rightarrow\infty}\inf_s I(\tilde{X}_n;\tilde{Y}_n|s)
\ea
where (b) holds if the channel is memoryless (not necessarily stationary or information-stable) and (c) holds if the channel is also stationary and $\tilde{X}_n$ is distributed according to $p_n(x)=\frac{1}{n}\sum_{k=1}^n p_{x_k}(x)$, where $\tilde{Y}_n$ is induced by $\tilde{X}_n$.
\end{prop}
\begin{proof}
See Appendix.
\end{proof}

Note that Proposition \ref{prop.uulineI<=liminf.inf.I} links the compound inf-information rate, whose definition does not include expectation, to the mutual information rate, i.e. an expected quantity, and (a) holds in full generality.
A sufficient condition to achieve the equality in (b) in \eqref{eq.uulineI<=liminf.inf.I} is well-known. Below, we obtain a sufficient condition for (a) to become equality.
\begin{prop}
\label{prop.uulineI=liminf.inf.I}
Consider a compound channel such that
\ba
\label{eq.prop.uulineI=liminf.inf.I.1}
\liminf_{n\rightarrow\infty} \inf_s \Pr\{|Z_{ns} - \uuline{I}(\bX,\bY)|> \delta\} = 0 \ \forall \delta>0
\ea
where $Z_{ns}=\frac{1}{n} i(X^n;Y^n|s)$, and at least one alphabet (input or/and output) is of finite cardinality. Then, its compound inf-information rate satisfies the following:
\ba
\label{eq.prop.uulineI=liminf.inf.I.2}
\uuline{I}(\bX,\bY) = \liminf_{n\rightarrow\infty} \inf_s \frac{1}{n}I(X^n;Y^n|s)
\ea
\end{prop}
\begin{proof}
See Appendix.
\end{proof}

\begin{remark}
Note that Proposition \ref{prop.uulineI=liminf.inf.I} holds even if the compound channel is information-unstable. Condition \eqref{eq.prop.uulineI=liminf.inf.I.1} means that there exists such sub-sequence $n_k$, $k=1...\infty$, and such channel states $s_k=s(n_k)$ that the sub-sequence of normalized information densities $Z_{n_k s_k}$ converges in probability to $\uuline{I}(\bX,\bY)$, i.e. that sub-sequence is information-stable.
\end{remark}

\begin{remark}
An equivalent to Proposition \ref{prop.uulineI=liminf.inf.I} is that
\ba
\exists \delta>0:\ \liminf_{n\rightarrow\infty} \inf_s \Pr\{|Z_{ns} - \uuline{I}(\bX,\bY)|> \delta\} > 0
\ea
is a necessary condition for the strict inequality in (a) in \eqref{eq.uulineI<=liminf.inf.I}, i.e. there exists no information-stable sub-sequence in the compound channel that would converge to $\uuline{I}(\bX,\bY)$.
\end{remark}

Next, let us establish a lower bound for the compound sup-information rate. Let
\ba
I_{n}(a) = \sup_s E\{Z_{ns} 1[Z_{ns} \le a]\}
\ea
and $I_n = \lim_{a\rightarrow\infty} I_n(a)$. Under the uniform (in $n$) convergence requirement for $I_n(a)\rightarrow I_n$, the following bound on the sup-information rate holds.

\begin{prop}
\label{prop.limsup.sup.I<=oolineI}
The following inequalities hold for the general compound channel:
\ba\notag
\label{eq.uuI<=liminfI<limsupI<=ooI}
\uuline{I}(\bX,\bY) &\le \liminf_{n\rightarrow\infty} \inf_s \frac{1}{n}I(X^n;Y^n|s)\\
&\le \limsup_{n\rightarrow\infty} \sup_s \frac{1}{n}I(X^n;Y^n|s)\\ \notag
 &\le \ooline{I}(\bX,\bY)
\ea
where the first two inequalities hold in full generality and the last inequality holds when the convergence $I_n(a)\rightarrow I_n$ as $a\rightarrow \infty$ is uniform in $n$. In particular, this holds when at least one alphabet is of finite cardinality.
\end{prop}
\begin{proof}
See Appendix.
\end{proof}

We are now in a position to establish the optimality of independent inputs for a compound memoryless (not necessarily stationary or information-stable) channel.

\begin{thm}[Optimality of Independent Inputs]
\label{thm.indep.input}
Consider a compound memoryless channel. Let $\bX$ and $\bY$ be its input and output sequences, and $\tilde{\bX}$, $\tilde{\bY}$ be sequences of independent symbols with the same per-symbol statistics as those of $\bX$ and $\bY$. Assume that
\ba
\label{eq.thm.sigmans=0}
\lim_{n\rightarrow\infty}\sup_ s \sigma_{ns}^2 =0
\ea
where $\sigma_{ns}^2$ is the variance of information density rate under independent inputs:
\ba
\sigma_{ns}^2 = \textsf{var} \bLF \frac{1}{n} \sum_{i=1}^{n} \ln \frac{p_s(\tilde{Y_i}|\tilde{X_i})}{p_s(\tilde{Y_i})} \bRF
\ea
Then,
\ba
\label{eq.thm.indep.3}
\uuline{I}(\bX;\bY) \le \uuline{I}(\tilde{\bX};\tilde{\bY})
\ea
i.e. independent signaling is optimal, and the compound channel capacity is
\ba
\label{eq.thm.indep.Cc}\notag
C_c &= \sup_{p(\tilde{\bx})} \uuline{I}(\tilde{\bX};\tilde{\bY})\\
 &= \liminf_{n\rightarrow\infty} \sup_{p(\tilde{x}^n)} \inf_ s \frac{1}{n} \sum_{k=1}^{n} I(\tilde{X}_k;\tilde{Y}_k|s)
\ea
where $I(X_k;Y_k|s)=E\{i(X_k;Y_k|s)\}$ is $k$-th symbol mutual information and $p(\tilde{x}^n)=\prod_{k=1}^n p_k(\tilde{x}_k)$ is memoryless input.
\end{thm}
\begin{proof}
In view of Proposition \ref{prop.uulineI<=liminf.inf.I}, the inequality in \eqref{eq.thm.indep.3}  is established by establishing
\ba
\uuline{I}(\tilde{\bX},\tilde{\bY}) = \liminf_{n\rightarrow\infty} \inf_s \frac{1}{n}\sum_{k=1}^n I(X_k;Y_k|s)
\ea
To see this, let $Z_{ns}=n^{-1}\sum_{k=1}^n i(\tilde{X}_k;\tilde{Y}_k|s)$ and apply Proposition \ref{prop.zns}. \eqref{eq.thm.indep.Cc} follows from \eqref{eq.thm.indep.3}.
\end{proof}

If, in addition, the channel is also stationary, then i.i.d. input is optimal and the familiar single-letter capacity expression results:
\bal
\label{eq.Cc.XY}
C_c=  \sup_{p(x)} \inf_ s I(X;Y|s).
\eal
Furthermore, since the uncertainty set $\sS$ is arbitrary, one can also treat the state $s$ as a sequence $s^n=\{s_1,..,s_n\}$ so that the memoryless channel model becomes
\bal\notag
p_{s^n}(y^n|x^n)= \prod_{k=1}^n p_{s_k}(y_k|x_k)
\eal
which is exactly the arbitrary varying channel (AVC)\footnote{This connection was pointed to us by Y. Steinberg.} \cite{Csiszar-92}\cite{Lapidoth-98A}. It follows from \eqref{eq.thm.indep.Cc} that its capacity $C_{AVC}$ is the same as the compound capacity in \eqref{eq.Cc.XY}, $C_c=C_{AVC}$, under the full CSI-R. Note that this result holds for deterministic coding and maximum as well as average error probability (recall that the AVC capacity can be different under average and maximum error probabilities, and also under deterministic and random coding; the AVC capacity under deterministic coding and maximum error probability is not known in general while some special cases have been settled \cite{Lapidoth-98A}\cite{Csiszar-11}). This extends the earlier result in \cite{Stambler-75} (established under average error probability) to the maximum error probability as well as to arbitrary input/output alphabets and channel state sets. It follows that allowing random (as opposed to deterministic) coding and/or average instead of maximum error probability does not increase the AVC capacity under the full CSI-R.

\begin{remark}
The condition in \eqref{eq.thm.sigmans=0} holds if any of the following holds:
\begin{enumerate}
    \item the variances of per-symbol information densities are uniformly bounded:
         \ba
         \sigma_{ks}^2 = \textsf{var}\{i(\tilde{X}_k;\tilde{Y}_k|s)\} \le A < \infty
         \ea
        which is the case if at least one alphabet is of finite cardinality (see Remark 3.1.1 in \cite{Han}, which is straightforward to extend to the compound setting);
    \item the per-symbol variances are bounded: $\sigma_{ks}^2 \le A_k < \infty$ and
        \ba
         \lim_{n\rightarrow\infty} \frac{1}{n^2} \sum_{k=1}^{n} A_k =0
        \ea
\end{enumerate}
\end{remark}

Let us now consider a $(n,\varepsilon_n,r_n)$-code for an arbitrary compound channel such that
\ba
\label{eq.code.1}
\lim_{n\rightarrow\infty} \varepsilon_n = 0, \lim_{n\rightarrow\infty} r_n = R
\ea
i.e. it achieves rate $R$ and arbitrary low error probability over that channel. What is the information density distribution (spectrum) induced by this code?

\begin{thm}
\label{thm.code}
Consider the code above operating on an arbitrary compound channel such that \eqref{eq.code.1} holds. If the input $X^n$ is uniformly distributed over the codewords, then the induced information density rate $n^{-1}i(X^n;Y^n|s)$ converges in probability to the code rate $R$ uniformly over the whole class of channels:
\ba
\label{eq.thm.code.1}
\lim_{n\rightarrow\infty} \sup_s \Pr\{|n^{-1}i(X^n;Y^n|s)-R| > \delta\}=0 \ \forall \delta>0
\ea
so that
\ba
\label{eq.thm.code.2}
\uuline{I}(\bX,\bY) = \ooline{I}(\bX,\bY) = \lim_{n\rightarrow\infty}  \frac{1}{n}I(X^n;Y^n|s)=R \ \forall s
\ea
\end{thm}
\begin{proof}
Since $R - \delta \le r_n \le R + \delta$ for any $\delta>0$ and sufficiently large $n$,
\ba \notag
\frac{1}{n}i(X^n;Y^n|s) &= \frac{1}{n}\ln \frac{p_s(X^n|Y^n)}{p(X^n)}\\ \notag
&\le \frac{1}{n}\ln \frac{1}{p(X^n)}\\
\label{eq.thm.code.3}
&= r_n \le R+\delta
\ea
where the last equality is due to $p(X^n)=1/M_n$. On the other hand, using Lemma \ref{lemma.comp.Verdu.Han},
\ba \notag
\varepsilon_{n} &\ge \sup_{s} \Pr\bLF n^{-1} i(X^n;Y^n|s) \le r_n - \delta \bRF - e^{-\delta n} \\ \notag
&\ge \sup_{s} \Pr\bLF n^{-1} i(X^n;Y^n|s) \le R - 2\delta \bRF - e^{-\delta n} \\ \notag
\ea
for any $\delta>0$, so that taking $\lim_{n\rightarrow\infty}$ on both sides, one obtains
\ba
\lim_{n\rightarrow\infty}\sup_{s} \Pr\bLF n^{-1} i(X^n;Y^n|s) \le R - 2\delta \bRF =0 \ \forall \delta>0
\ea
Combining this with \eqref{eq.thm.code.3}, \eqref{eq.thm.code.1} follows. To prove \eqref{eq.thm.code.2}, note that 1st equality follows from \eqref{eq.thm.code.1} and 2nd equality (and the existence of corresponding limit) follows from \eqref{eq.uuI<=liminfI<limsupI<=ooI}.
\end{proof}

Theorem \ref{thm.code} generalizes Theorem 3.2.3 in \cite{Han}\footnote{this theorem has appeared before, albeit in a different form, in \cite{Ting}.} to the compound channel setting and the convergence in probability holds for the whole class of channels  uniformly in $s$, not just for each channel individually. Even though the compound channel is allowed to be information-unstable, the code-induced information density is information-stable and the corresponding information spectrum is a single atom equal to the code rate and also the mutual information rate under any channel state in the uncertainty set (so that the mutual information rate is state-independent), as long as (i) the error probability converges to zero, and (ii) the sequence of code rates converges. In a sense, this constitutes a link between information-unstable (non-ergodic, non-stationary) and information-stable regimes of a compound channel. Combining Theorem \ref{thm.code} with Lemma \ref{lemma.comp.Feinstein}, one concludes that information stability over a compound channel is both necessary and sufficient for a code in \eqref{eq.code.1} to exist.

\section{Strong Converse for the General Compound Channel}
\label{sec.strong.conv}

Strong converse ensures that a slightly larger error probability cannot be traded off for a higher data rate (since the transition from arbitrary low to high error probability is sharp). Another motivation is to consider a scenario where a capacity-achieving code is designed for a given SNR and the actual system SNR drops below this value so that the system operates at a rate above the channel capacity. If the strong converse holds, this results in large error rate while only gradual degradation occurs otherwise. A formal definition follows.

\begin{defn}
A compound channel is said to satisfy strong converse if
\ba
\label{eq.strong.conv.d1}
\lim_{n\rightarrow\infty} \varepsilon_n = 1
\ea
for any code satisfying
\ba
\label{eq.strong.conv.d2}
\liminf_{n\rightarrow\infty} r_n > C_c
\ea
\end{defn}

To obtain conditions for strong converse, let $\check{I}(\bX;\bY)$ be the "worst-case" sup-information rate,
\ba
\label{eq.tilde{I}}
\check{I}(\bX;\bY) \triangleq \inf_{R}\bLF R: \lim_{n\rightarrow\infty} \inf_{s\in \sS} \Pr\left\{Z_{ns} > R \right\} =0 \bRF
\ea
where $Z_{ns} = n^{-1}i(X^n;Y^n|s)$ is the information density rate, and $I_{ns}(a)$ be the truncated mutual information,
\ba
I_{ns}(a) \triangleq E\{Z_{ns}1[Z_{ns} \le a]\},\ I_{ns} = \lim_{a\rightarrow\infty} I_{ns}(a)
\ea
where $1[\cdot]$ is the indicator function and $I_{ns}=I(X^n;Y^n|s)$ is the mutual information under channel state $s$. The sup-information rate $\bar{I}(\bX;\bY|s)$ under channel state $s$ is defined as
\ba
\bar{I}(\bX;\bY|s) \triangleq \inf_{R}\bLF R: \lim_{n\rightarrow\infty} \Pr\left\{Z_{ns} \ge R \right\} =0 \bRF
\ea

Fig. \ref{fig:Figure_I_hat} illustrates various information rates for a two-state channel.

\begin{figure}[htbp]
\centerline{\includegraphics[width=2in]{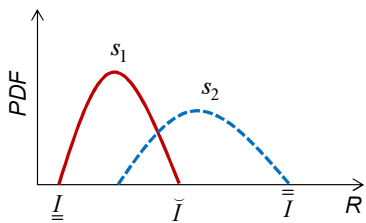}} \caption{An illustration of the information rates $\uuline{I}$, $\check{I}$ and $\ooline{I}$ for a two-state channel. Solid and dashed lines indicate the asymptotic distributions of the information density rate $n^{-1}i(X^n;Y^n|s)$ under the two states $s_1$ and $s_2$.}
\label{fig:Figure_I_hat}
\end{figure}

The following Proposition establishes an ordering of various information rates.

\begin{prop}
\label{prop.ineq.tildeI}
The following inequalities hold for any input
\ba \notag
\uuline{I}(\bX;\bY) &\le \check{I}(\bX;\bY)\\ \notag
 &\le \inf_s \bar{I}(\bX;\bY|s)\\ \notag
  &\le \sup_s \bar{I}(\bX;\bY|s)\\
  &\le \ooline{I}(\bX;\bY)
\ea
In addition,
\ba
\label{eq.uulineI<=liminf.inf.I.2}
\uuline{I}(\bX,\bY) \le \liminf_{n\rightarrow\infty} \inf_s \frac{1}{n}I(X^n;Y^n|s) \le \check{I}(\bX;\bY)
\ea
where the 2nd inequality holds if the convergence in $I_{ns}(a) \rightarrow I_{ns}$ is uniform.
\end{prop}
\begin{proof}
see the Appendix.
\end{proof}

It can be shown, via examples, that all inequalities can be strict. Using this Proposition, sufficient and necessary conditions for the strong converse to hold can be established.

\begin{thm}
\label{thm.strong.conv}
A sufficient and necessary condition for the general compound channel to satisfy strong converse is
\ba
\label{eq.strong.conv.1}
\sup_{p(\bx)} \uuline{I}(\bX;\bY) = \sup_{p(\bx)} \check{I}(\bX;\bY)
\ea
If this holds and the convergence $I_{ns}(a)\rightarrow I_{ns}$ is uniform in $n, s$ for any input $\bX^*$ satisfying $\uuline{I}(\bX^*;\bY^*)> C_c - \delta$ for some $\delta>0$ (i.e. the input $\bX^*$ is $\delta$-suboptimal), then
\ba
\label{eq.strong.conv.1a}
C_c =\sup_{p(\bx)} \check{I}(\bX;\bY) = \liminf_{n\rightarrow\infty} \sup_{p(x^n)} \inf_s \frac{1}{n} I(X^n;Y^n|s)
\ea
The condition \eqref{eq.strong.conv.1} is equivalent to any of the following:

1) for any $\delta>0$ and any input $\bX^*$ satisfying $\uuline{I}(\bX^*;\bY^*)> C_c - \delta$,
\ba
\label{eq.strong.conv.2}
\lim_{n\rightarrow\infty} \inf_s \Pr\{|Z_{ns}^* - C_c|> \delta\} = 0
\ea
where $Z_{ns}^*=\frac{1}{n} i({X^n}^*;{Y^n}^*|s)$ is the information density rate under input $\bX^*$.

2) for any input $\bX$ and any $\delta>0$,
\ba
\label{eq.strong.conv.2a}
\lim_{n\rightarrow\infty} \inf_s \Pr\{Z_{n s} > C_c + \delta\} = 0
\ea
\end{thm}

\begin{proof}
see the Appendix.
\end{proof}

\begin{remark}
In the case of a single-state channel,
\ba
\uuline{I}(\bX;\bY) = \underline{I}(\bX;\bY),\ \check{I}(\bX;\bY) = \overline{I}(\bX;\bY)
\ea
where $\underline{I}(\bX;\bY),\ \overline{I}(\bX;\bY)$ are inf and sup-information rates for the regular (single-state) channel, and Theorem \ref{thm.strong.conv} reduces to the corresponding Theorem in \cite{Verdu}\cite{Han}.
\end{remark}

\begin{remark}
Note that, under the conditions of Theorem \ref{thm.strong.conv} that lead to \eqref{eq.strong.conv.1a}, the compound channel behaves \textit{ergodically} (the mutual information has operational meaning) even though \textit{no assumption of ergodicity} or information stability was made upfront.
\end{remark}

Below, we consider a special case when the supremum in \eqref{eq.strong.conv.1} is achieved.

\begin{cor}
\label{cor.strong.conv.cor.1}
If the channel satisfies strong converse and the supremum in $\sup_{p(\bx)} \uuline{I}(\bX;\bY)$ is achieved, i.e.
\ba
\exists \bX^*:\ \uuline{I}(\bX^*;\bY^*) = C_c
\ea
then $\check{I}(\bX^*;\bY^*) = C_c$ and there exists such sequence of channel states $s(n)$ that the corresponding sequence of normalized information densities $Z_{n s(n)}^*$ (under input $\bX^*$) converges in probability to the compound channel capacity $C_c$,
\ba
\label{eq.strong.conv.cor.1}
\lim_{n\rightarrow\infty} \Pr\{|Z_{n s(n)}^* - C_c|> \delta\} = 0 \ \forall \delta>0
\ea
i.e. this sequence (which represents worst-case channels in the uncertainty set) is information-stable.
\end{cor}
\begin{proof}
Observe that $\uuline{I}(\bX^*;\bY^*) = C_c$ implies
\ba
C_c = \uuline{I}(\bX^*;\bY^*) \le \check{I}(\bX^*;\bY^*) \le \sup_{p(\bx)} \check{I}(\bX;\bY) = C_c
\ea
so that $\check{I}(\bX^*;\bY^*) = C_c$ follows, which also implies that
\ba
\lim_{n\rightarrow\infty} \inf_{s} \Pr\bLF Z_{n s}^* > C_c + \delta \bRF =0\ \forall \ \delta>0
\ea
On the other hand, $\uuline{I}(\bX^*;\bY^*) = C_c$ implies
\ba
\lim_{n\rightarrow\infty} \sup_{s} \Pr\bLF Z_{n s}^* < C_c - \delta \bRF =0\ \forall \ \delta>0
\ea
and hence
\ba
\label{eq.strong.conv.cor.5}
\lim_{n\rightarrow\infty} \inf_s \Pr\{|Z_{n s}^* - C_c|> \delta\} = 0 \ \forall \delta>0
\ea
follows. Next, we need the following technical Lemma.

\begin{lemma}
\label{lemma.x_ns}
Let $\{x_{ns}\}$ be a non-negative compound sequence such that
\ba
\label{eq.Lemma.x_ns.1}
\lim_{n\rightarrow\infty} \inf_{s} x_{ns} = 0
\ea
Then, there exists such sequence of states $s(n)$ that
\ba
\label{eq.Lemma.x_ns.2}
\lim_{n\rightarrow\infty} x_{ns(n)} = 0
\ea
\end{lemma}
\begin{proof}
When $\inf_s$ is achieved, the statement is trivial. To prove it in the general case, observe that, from the definition of $\inf_s$ and for any $n$, there always exists such $s(n)$ that
\bal
x_{ns(n)} < \inf_s x_{ns} +1/n
\eal
so that taking $\lim_{n\rightarrow \infty}$ of both sides, one obtains \eqref{eq.Lemma.x_ns.2}\footnote{this way of proof was suggested by a reviewer.}.
\end{proof}

Using this Lemma, \eqref{eq.strong.conv.cor.5} implies the existence of a sequence of channel states $s(n)$ such that \eqref{eq.strong.conv.cor.1} holds.
\end{proof}

\begin{remark}
Note that, under the conditions of Corollary \ref{cor.strong.conv.cor.1}, the sequence $s(n)$ of worst-case channel states is information-stable even though no assumption of information stability was made upfront.
\end{remark}

\begin{remark}
In light of Lemma \ref{lemma.x_ns}, condition \eqref{eq.strong.conv.2a} means that there exists such sequence of (bad) channel states $s(n)$ that the information spectrum of the corresponding sequence of normalized information densities $Z_{n s(n)}$  does not exceed $C_c$ under any input, i.e.
\ba
\exists s(n):\ \lim_{n\rightarrow\infty} \Pr\{Z_{n s(n)}> C_c + \delta\} = 0 \ \forall \delta>0
\ea
\end{remark}

\section{$\varepsilon$-Capacity of Compound Channels}
\label{sec.C_eps}

Let us now consider the so-called $\varepsilon$-channel capacity, where the error probability is not required to be arbitrary small but rather to be not larger than a given value $\varepsilon$ asymptotically. $(n,r_n,\varepsilon_n)$-code over a compound channel is defined in the same way as before. $\varepsilon$-achievable rate and capacity are defined as in \cite{Verdu}\cite{Han} (for the non-compound setting), where the extension to the compound setting follows from \eqref{eq.e_n.comp} and the requirement of codewords to be independent of channel state.

\begin{defn}
Rate $R$ is $\varepsilon$-achievable over a compound channel if there exists $(n,r_n,\varepsilon_n)$-code (where codewords are independent of channel state) such that
\ba
\limsup_{n\rightarrow\infty} \varepsilon_n \le \varepsilon, \ \liminf_{n\rightarrow\infty} r_n \ge R
\ea
\end{defn}

\begin{defn}
$\varepsilon$-capacity $C_{\varepsilon}$ of a compound channel is the largest $\varepsilon$-achievable rate over that channel:
\ba
C_{\varepsilon}= \sup\{R: R\ \mbox{is $\varepsilon$-achievable} \}
\ea
\end{defn}

To characterise $C_{\varepsilon}$ of the general compound channel, let us introduce the following quantities:
\ba
&F_{\bX}(R) \triangleq \limsup_{n\rightarrow\infty}\sup_s \Pr\bLF \frac{1}{n} i(X^n;Y^n|s) \le R \bRF \\
&\uuline{I}_{\varepsilon}(\bX;\bY) \triangleq \sup\{R: F_{\bX}(R) \le \varepsilon\}
\ea
Roughly speaking, $F_{\bX}(R)$ is the asymptotic CDF of information density rate of the compound channel and, as will be shown below, $\uuline{I}_{\varepsilon}(\bX;\bY)$ is $\varepsilon$-achievable rate over that channel. Its $\varepsilon$-capacity is as follows.

\begin{thm}
\label{thm.Ce.general}
Consider the general compound channel where channel state $s\in\sS$ is independent of the input and is known to the receiver; the transmitter  knows only the (arbitrary) uncertainty set $\sS$. Its $\varepsilon$-capacity is
\ba
C_{\varepsilon}= \sup_{p(\bx)} \uuline{I}_{\varepsilon}(\bX;\bY)
\ea
\end{thm}
\begin{proof}
The proof follows the steps of that of Theorem \ref{thm.C.general}. First, fix $p(\bx)$ and set $r_n \le \uuline{I}_{\varepsilon}(\bX;\bY) - 2 \gamma$. From Lemma \ref{lemma.comp.Feinstein}, one obtains a code such that
\ba \notag
\limsup_{n\rightarrow\infty} \varepsilon_{n} &\le \limsup_{n\rightarrow\infty} \sup_{s\in \sS} \Pr\bLF Z_{ns} \le \uuline{I}_{\varepsilon}(\bX;\bY) - \gamma \bRF \\
&= F_{\bX}(\uuline{I}_{\varepsilon}(\bX;\bY) - \gamma) \le \varepsilon
\ea
so that $\uuline{I}_{\varepsilon}(\bX;\bY) - 2\gamma$ is achievable for any $\gamma>0$, from which one obtains $C_{\varepsilon} \ge \sup_{p(\bx)} \uuline{I}_{\varepsilon}(\bX;\bY)$.

Next, let $R=\sup_{p(\bx)} \uuline{I}_{\varepsilon}(\bX;\bY)$ and set $r_n \ge R + 2 \gamma$ and use Lemma \ref{lemma.comp.Verdu.Han} to obtain
\ba \notag
\limsup_{n\rightarrow\infty}\varepsilon_{n} &\ge \limsup_{n\rightarrow\infty}\sup_{s\in \sS} \Pr\bLF Z_{ns} \le R + \gamma \bRF\\ \notag
&\ge \limsup_{n\rightarrow\infty}\sup_{s\in \sS} \Pr\bLF Z_{ns} \le \uuline{I}_{\varepsilon}(\bX;\bY) + \gamma \bRF\\
&= F_{\bX}(\uuline{I}_{\varepsilon}(\bX;\bY) + \gamma) > \varepsilon
\ea
where the last inequality follows from the definition of $\uuline{I}_{\varepsilon}(\bX;\bY)$, so that no rate above $R$ is $\varepsilon$-achievable and hence $C_{\varepsilon} \le \sup_{p(\bx)} \uuline{I}_{\varepsilon}(\bX;\bY)$.
\end{proof}

Similarly to the previous section, one can exploit the uniform convergence property and extend Theorem \ref{thm.C.uniform} to $\varepsilon$-capacity. To this end, let
\ba
\label{eq.F_XsR}
F_{\bX}(R,s) \triangleq \limsup_{n\rightarrow\infty} \Pr\bLF \frac{1}{n} i(X^n;Y^n|s)\le R \bRF
\ea
and define the $\varepsilon$-inf-information rate for channel state $s$:
\ba
\underline{I}_{\varepsilon}(\bX;\bY|s) \triangleq \sup\{R: F_{\bX}(R,s) \le \varepsilon\}
\ea

\begin{defn}
\label{defn.eps-uniform.ch}
Let $\bX_{\delta}$ be a $\delta$-suboptimal input so that $\uuline{I}_{\varepsilon}(\bX_{\delta};\bY_{\delta})\ge C_{\varepsilon} -\delta$.
A compound channel is $\varepsilon$-uniform if there exists $\delta\ge 0$ such that, for any $\bX_{\delta}$ and any rate $R$ such that $C_{\varepsilon} -2\delta \le R \le C_{\varepsilon} +2\delta$, the convergence to the limit in \eqref{eq.F_XsR} is uniform in $s\in \sS$ for any $\delta$-suboptimal input, $\bX=\bX_{\delta}$.
\end{defn}

It is straightforward to see that any finite-state channel is $\varepsilon$-uniform under any input. Following the steps of the previous section, one obtains the following bound which results in the familiar $\sup-\inf$ capacity formula.

\begin{prop}
\label{prop.uulineIe<=underlineIe}
The following inequality holds for a general compound channel:
\ba
\label{eq.uulineIe<=underlineIe}
\uuline{I}_{\varepsilon}(\bX,\bY) \le  \underline{I}_{\varepsilon}(\bX,\bY) \triangleq \inf_s \underline{I}_{\varepsilon}(\bX,\bY|s)
\ea
with equality in the inequality for an $\varepsilon$-uniform compound channel under any $\delta$-suboptimal input, $\bX=\bX_{\delta}$.
\end{prop}
\begin{proof}
see Appendix.
\end{proof}

Using Proposition \ref{prop.uulineIe<=underlineIe}, the $\varepsilon$-capacity of an $\varepsilon$-uniform compound channel can be expressed using the familiar $\sup-\inf$ expression.

\begin{thm}
\label{thm.Ce.uniform}
Consider the general compound channel where the channel state $s \in \mathcal{S}$ is known to the receiver but not the transmitter and is independent of the channel input; the transmitter knows the (arbitrary) uncertainty set $\mathcal{S}$. Its compound  $\varepsilon$-capacity is bounded by
\ba
C_{\varepsilon} \le \sup_{p(\bx)} \inf_{s\in \sS} \underline{I}_{\varepsilon}(\bX;\bY|s)
\ea
with equality for an $\varepsilon$-uniform compound channel. In particular, this holds when $\sS$ is of finite cardinality.
\end{thm}

\section{Mixed and Composite Channels}
\label{sec.Mixed.Composite}

Let us consider a mixed channel of the form:
\ba
p(y^n|x^n)= \sum_{s=1}^{\infty} \alpha_s p_s(y^n|x^n)
\ea
where $\alpha_s \ge 0$, $s=1,2,...$, $\sum_s \alpha_s =1$, which is a mixture of individual channel states. The capacity of this channel in the general case (e.g. information-unstable) was found in \cite{Han}:
\ba
\label{eq.Cmix}
C_{mix} = \sup_{p(\bx)} \inf_{s: \alpha_s >0} \underline{I} (\bX;\bY|s)
\ea
where $\underline{I} (\bX;\bY|s)$ in the inf-information rate induced by $p_s(y^n|x^n)$. Following Proposition \ref{prop.uuline(I)<=underline(I)}, the compound channel capacity is upper bounded by the mixed channel capacity:
\ba
\label{eq.Cc<=Cmix}
C_c = \sup_{p(\bx)} \uuline{I}(\bX;\bY) \le C_{mix}
\ea
where the compound channel state set $\sS=\{s: \alpha_s >0\}$. As the examples in the next Section demonstrate, the inequality can be strict. Comparing \eqref{eq.Cmix} to Theorem \ref{thm.C.uniform}, one concludes that \eqref{eq.Cc<=Cmix} holds with equality provided that the compound channel is uniform (which holds if $\sS$ is of finite cardinality).

Composite channels have been introduced and studied in \cite{Effros}. This type of channels is similar to compound channels except that there is a probability measure associated with each channel state: $\{\alpha_s, p_s(y^n|x^n)\}$. A channel state $p_s(y^n|x^n)$ is selected with probability $\alpha_s$ and kept constant during the whole transmission. Since the channel description is entirely probabilistic, the general formula in \cite{Verdu} applies and its capacity is the same as the mixed channel capacity in \eqref{eq.Cmix}: $C_{com}=C_{mix}$, and the inequality in \eqref{eq.Cc<=Cmix} applies.

\section{Examples}
\label{sec.Examples}

\subsection{Example 1}

To demonstrate the difference between Theorems \ref{thm.C.general} and \ref{thm.C.uniform} and the fact that inequality in \eqref{eq.prop.uuline(I)<=underline(I)} can be strict, consider the following binary non-stationary channel with memory:
\ba
p_s(y^n|x^n)=p_s(y^n) \ \mbox{if}\ n \le s
\ea
i.e. the output is independent of the input. If $n>s$, then the channel is $n$-th extension of BSC with zero cross-over probability, and $\sS = \{1,2,...\}$. This can model a channel with memory where the noise coherence time $\tau = s$ so that blocklength $n>\tau$ is required to achieve low error probability. Since $i(X^n;Y^n|s)=0$ if $s\ge n$ , it follows that $\uuline{I}(\bX;\bY)=0$ while $\underline{I}(\bX;\bY|s)=\ln 2\ \forall s$ under i.i.d. equiprobable input, so that
\ba
\uuline{I}(\bX;\bY)=0 < \underline{I}(\bX;\bY) = \inf_s \underline{I}(\bX;\bY|s) =\ln 2
\ea
and hence
\ba
C_c = \sup_{p(\bx)} \uuline{I}(\bX;\bY)=0 < \ln 2 = \sup_{p(\bx)} \inf_{s\in \sS} \underline{I}(\bX;\bY|s)
\ea
The compound capacity $C_c$ is zero because for any blocklength, does not matter how large, there are always channel states with error probability close to 1 so that arbitrary low error probability is not attainable. The standard $\sup-\inf$ expression falls short of the channel capacity in this case because this compound channel is not uniform. It also demonstrates that Theorem 3.3.5 in \cite{Han} cannot ensure reliable communications for infinite-state compound channels. Note that if the coherence time becomes bounded, i.e. $\tau=s\le S < \infty$, then $C_c = \sup_{p(\bx)} \inf_{s\le S} \underline{I}(\bX;\bY|s) = \ln 2$ as one can use sufficiently-long codewords constructed for memoryless BSC (notice also that the channel becomes uniform in this case).

This example can be extended to a scenario where the channel is $\mathrm{BSC}(q_1)$ if $n \le s$ and $\mathrm{BSC}(q_2)$ otherwise, where $\mathrm{BSC}(q)$ is the $n$-th extension of a binary symmetric channel with crossover probability $q$, $0\le q_2 < q_1 \le 1/2$, so that
\ba\notag
C_c &= \ln 2-H(q_1)\\
&< \ln 2 -H(q_2)\\ \notag
 &= \sup_{p(\bx)} \inf_{s\in \sS} \underline{I}(\bX;\bY|s)
\ea
where $H(q)$ is the binary entropy function.

\subsection{Example 2}

Let us consider the following additive noise compound channel model:
\ba
Y_k = X_k + Z_{ks}
\ea
where $k$ is (discrete) time index, $s$ is a state, the compound noise process $\{Z_{ks}\}_{k=1}^{\infty}$ is arbitrary but independent of $\{X_{k}\}_{k=1}^{\infty}$, and all alphabets are binary. Using Theorem \ref{thm.C.general}, its compound channel capacity can be evaluated via the properties in Proposition \ref{prop.properties.I}:
\ba
\label{eq.Example2.2}
C_c = \sup_{p(\bx)} \uuline{I} (\bX;\bY) = \ln2 - \ooline{H}(\bZ)
\ea
To see this, observe that
\ba
\label{eq.Example2.3} \notag
\uuline{H}(\bY) - \ooline{H}(\bZ) &\le \uuline{I} (\bX;\bY)\\
&\le \ooline{H}(\bY) - \ooline{H}(\bZ)\\  \notag
&\le \ln 2 - \ooline{H}(\bZ)
\ea
since $\ooline{H}(\bY|\bX)=\ooline{H}(\bZ)$. On the other hand,
\ba
\ln2 \ge \uuline{H}(\bY) \ge \uuline{H}(\bY|\bZ) = \uuline{H}(\bX)
\ea
and likewise for the sup-entropy rates. Using i.i.d. equiprobable sequence for $\bX$ results in $\ooline{H}(\bY)= \uuline{H}(\bY) = \uuline{H}(\bX) =\ln 2$ and thus the lower and upper bounds in \eqref{eq.Example2.3} coincide resulting in \eqref{eq.Example2.2} (this also shows that i.i.d. equiprobable signaling is optimal regardless of the statistics of the noise).

When there is only one channel state (i.e. non-compound channel), the capacity was obtained before in \cite{Verdu} using the general formula there:
\ba
\label{eq.C.Zn}
C = \sup_{p(\bx)} \underline{I} (\bX;\bY) = \ln2 - \overline{H}(\bZ)
\ea
While the two expressions look remarkably similar, they may produce significantly different results. To see this, consider the following compound noise process:
\ba
\label{eq.Zsn}
Z_s^n = \{w_1,w_2,...w_s,0,0...0\}
\ea
i.e. for a given state $s$, first $s$ symbols are i.i.d. equiprobable binary random variables $w_1...w_s$ and the last $n-s$ symbols are zeros. The associated probability distribution $p_s(z^n)=1/2^n$ if $s\ge n$ so that $\ooline{H}(\bZ)=\ln 2$ and $C_c = 0$. This result can be explained by observing that for any $n$, does not matter how large, there are always channel states $s\ge n$ for which the channel is BSC(1/2), i.e. useless. On the other hand, using \eqref{eq.C.Zn} for any channel state $s$ results in
\ba
C_s = \sup_{p(\bx)} \underline{I} (\bX;\bY|s) = \ln 2 - \overline{H}(\bZ|s) = \ln 2
\ea
since, as it can be easily demonstrated, $\overline{H}(\bZ|s)=0$ for any $s$ (loosely speaking, this is because the random part of the sequence in \eqref{eq.Zsn} is negligible when $n \rightarrow \infty$). If one attempts to use Theorem \ref{thm.C.uniform} (or, equivalently, Theorem 3.3.5 in \cite{Han}),
\ba
\sup_{p(\bx)} \inf_{s} \underline{I}(\bX;\bY|s) = \ln 2 = C_s > C_c =0
\ea
since, as can be easily seen, $\underline{I}(\bX;\bY|s) = \ln 2$ when the input is i.i.d. equiprobable. The discrepancy is explained by the fact that this compound channel is not uniform and thus Theorem \ref{thm.C.uniform} and Theorem 3.3.5 in \cite{Han} do not apply.


\subsection{Example 3}
To demonstrate the practical utility of Theorems \ref{thm.C.general}, \ref{thm.C.uniform}, let us consider the following discrete-time wireless channel model:
\ba
y_i = h x_i +\xi_i
\ea
where $h$ is the channel gain, $\xi$ is the noise of variance $\sigma_{\xi}^2$, and $i$ is discrete time. The channel is memoryless. The channel gain $h$ models the wireless propagation path loss from the Tx to the Rx. Noise $\xi$ models thermal noise as well as external (e.g. multi-user) interference.

First, assume that $h$ is a given (fixed) constant known to the Tx and Rx. Further assume that $\sigma_{\xi}$ is randomly selected at the beginning and held constant during  the transmission, so that $\sigma_{\xi} = \sigma_1$ with probability $p_1>0$ and $\sigma_{\xi} = \sigma_2$ with probability $p_2 = 1-p_1$, $\sigma_1 > \sigma_2$. This can model a scenario where interference (from another user) is present with probability $p_1$ and absent with probability $p_2$, so that $\sigma_2^2 = \sigma_0^2$, $\sigma_1^2 = \sigma_0^2 + \sigma_I^2$, where $\sigma_{0(I)}^2$ is the noise (interference) power. Clearly, the channel is non-ergodic (information-unstable) so that
\ba
\frac{1}{n} i(X^n;Y^n|h) &\rightarrow I_{\bx}(h,\sigma_{\xi})
\ea
where $I_{\bx}(h,\sigma_{\xi})$ is the mutual information rate for given $h$, $\sigma_{\xi}$ and $p(\bx)$. Since $\sigma_{\xi}$ is random, so is $I_{\bx}(h,\sigma_{\xi})$ and thus $\frac{1}{n} i(X^n;Y^n|h)$ converges to $I_{\bx}(h,\sigma_{k})$ with probability $p_k$, $k=1,2$. The largest achievable rate under given $p(\bx)$ and arbitrary-small error probability is
\ba
R = \underline{I}(\bX;\bY|h) = I_{\bx}(h,\sigma_{1}) < I_{\bx}(h)
\ea
where $I_{\bx}(h) = p_1 I_{\bx}(h,\sigma_{1}) + p_2 I_{\bx}(h,\sigma_{2})$ is the regular mutual information rate, i.e. falls short of the mutual information rate (since the channel is information-unstable), where we assumed that $I_{\bx}(h,\sigma)$ is decreasing in $\sigma$. The difference can be significant if the noise power is large enough.

Now assume that $h$ is not known to the Tx but is known to belong to the uncertainty set $\mathcal{S} = [h_1, h_2]$, $0 \le h_1 < h_2$ (e.g. due to uncertainty in the user location, which affects the propagation path loss), so that a single code has to be designed to operate on all such channels. It can be seen that this compound channel is uniform. The compound capacity of this information-unstable channel is
\ba \notag
C &= \sup_{p(\bx)} \inf_h \underline{I}(\bX;\bY|h)\\
 &= \sup_{p(\bx)} I_{\bx}(h_1, \sigma_{1})\\  \notag
 &< \sup_{p(\bx)} I_{\bx}(h_1)
\ea
i.e. falls short of the regular compound channel capacity (which would be the capacity if the channel were information-stable).

It is clear that this example also extends to the case of any number of possible levels of $\sigma_{\xi}$ or when $\sigma_{\xi}$ is a continuous random variable characterized by the density $f(\sigma)$, in which case $\sigma_1 = \sup\{\sigma: f(\sigma)>0\}$ is the supremum of the support set of $\sigma_{\xi}$. A compound channel with memory can be considered in a similar way.

\subsection{Example 4: the impact of the Rx CSI}
\label{sec.Rx.CSI}

All the results in this paper are based on the assumption of the full Rx CSI. A question arises as to whether some of these results hold if this assumption is removed. The following example from \cite{Lapidoth-98B} demonstrates that the key result in Theorem \ref{thm.C.general} does not hold in general without such assumption.

Consider the following compound channel, which is binary, deterministic and fixed in time:
\bal
y_k = x_k + \theta_k
\eal
where $k$ is discrete time and the state $s$ is defined from
\bal
s = \sum_{i=1}^{\infty} 2^{-i} \theta_i,\ 0\le s \le 1,
\eal
i.e. $\theta_i$ is $i$-th binary digit of $s$. It is straightforward to verify that, for each channel state, this channel is information-stable for each $s$ and, for the uniform input $p(x^n)=1/2^n$,
\bal\notag
n^{-1}i(X^n;Y^n|s)\overset{\Pr}{=} \ln 2,\ \underline{I}(\bX;\bY|s)=\ln 2,\ \uuline{I}(\bX;\bY)=\ln 2,
\eal
i.e. this is a uniform compound channel, and
\bal
\sup_{p(\bx)} \uuline{I}(\bX;\bY)=\ln 2
\eal
Yet, with no Rx CSI, the capacity of this compound channel is $C_c=0$ \cite{Lapidoth-98B}. This can be easily established by observing that this is a binary discrete memoryless channel in disguise, which is required to work for every possible (and unknown) noise sequence and hence the same strategy can be used for the binary symmetric channel with cross-over probability of 1/2, for which the capacity is zero. Hence, Theorem \ref{thm.C.general} does not hold for this channel under no Rx CSI. This example also shows that Theorem 3.3.5 in \cite{Han} does not hold in general for infinite-state channels.

\section{Conclusion}

The general formula for the compound channel capacity with full CSI-R has been established using the information density approach, which does not require the channel to be stationary, ergodic, or information-stable, and which applies to any channel uncertainty set (not only countable or finite-state). The conditions for the worst-case and compound capacities to be equal are given. The compound inf-information rate plays a key role for the general formula. Its properties are studied, including the data processing inequality and optimality of independent inputs for the general compound memoryless channel. As a by-product, the AVC capacity is established under deterministic code and maximum error probability. The $\varepsilon$-capacity of the general compound channel is established and the sufficient and necessary conditions for the strong converse to hold are given.

Examples are provided, which show that finite and infinite-state compound channels can behave differently and which demonstrate the utility of the results in wireless communications.

\section{Acknowledgement}
The authors are grateful to S. Verdu and E. Telatar for insightful discussions and suggestions, and A. Lapidoth for valuable comments.

\section{Appendix}

\subsection{Proof of Lemma \ref{lemma.comp.Feinstein}}

Let us define
\ba
B_s(x^n)=\{y^n: i(x^n;y^n|s) \ge \ln \alpha\}, \ \alpha=M_n e^{n\gamma}, \\
\label{eq.lambda_n}
\lambda_n =\sup_{s\in \sS} \Pr\bLF i(X^n;Y^n|s) \le \ln \alpha \bRF + M_n/\alpha
\ea
and observe, for future use, that
\ba
\notag
1 &\ge \Pr \bLF Y^n \in B_s(x^n)| x^n \bRF \\ \notag
&= \sum_{y^n\in B_s(x^n)} p_s(y^n|x^n) \\ \notag
& \overset{(a)}{\ge} \alpha \sum_{y^n\in B_s(x^n)} p_s(y^n) \\
&= \alpha P_s(B_s(x^n))
\ea
from which it follows that
\ba
\label{eq.Pr.Bs<1/a}
P_s(B_s(x^n)) \le 1/\alpha \ \forall s, x^n,
\ea
where (a) follows from $p_s(y^n|x^n) \ge \alpha p_s(y^n)$ $\forall y^n \in B_s(x^n)$.

We use an iterative codebook construction similar to that in Section 3.5 of \cite{Ash} but properly extended to the compound channel setting here. Fix the input distribution $p(\bx)$. Find $x^n$ such that
\ba
x^n:\ \inf_s P_s(B_s(x^n)|x^n) \ge 1-\lambda_n
\ea
and use it as codeword 1, $\bu_1=x^n$ (note that this codeword is independent of channel state $s$); set the decision region $D_{1s}=B_s(\bu_1)$ for this codeword, so that probability of correct decision for this codeword is at least $1-\lambda_n$.

Next, find $x^n \neq \bu_1$ such that
\ba
x^n:\ \inf_s P_s(B_s(x^n)-D_{1s}|x^n) \ge 1-\lambda_n
\ea
and use it as codeword 2, $\bu_2=x^n$; set the decision region $D_{2s}=B_s(\bu_2)-D_{1s}$.

For codeword $K$, find $x^n \neq \bu_k, k=1...K-1$, such that
\ba
x^n:\ \inf_s P_s\left(B_s(x^n)-\bigcup_{k=1}^{K-1}D_{ks}|x^n\right) \ge 1-\lambda_n
\ea
and set $\bu_K=x^n$, $D_{Ks}=B_s(\bu_K)-\bigcup_{k=1}^{K-1}D_{ks}$.

Assume that the process stops at $k=K$, i.e. no further $x^n$ can be found satisfying the required inequality, so that:
\ba
\label{eq.K}
\inf_s P_s\left(B_s(x^n)-D_s|x^n\right) < 1-\lambda_n \ \forall x^n \neq \bu_k, k=1...K.
\ea
where $D_s=\bigcup_{k=1}^{K}D_{ks}$. The same inequality also holds for $x^n=\bu_k$, since
\ba
\label{eq.K.uk}
B_s(\bu_k)-D_s = B_s(\bu_k)-\bigcup_{l=1}^{K}B_s(\bu_l) = \emptyset
\ea
The following Lemma shows that a sufficiently large number of codewords can be constructed in this way.

\begin{lemma}
\label{lemma.K>Mn}
The algorithm above generates $K>M_n$ codewords.
\end{lemma}
\begin{proof}
To see this, observe that it follows from \eqref{eq.K} and \eqref{eq.K.uk} that there exists such channel state $s_0$ that
\ba
P_s\left(B_s(x^n)-D_s|x^n\right) < 1-\lambda_n \ \forall x^n, s=s_0
\ea
For this channel state, one obtains:
\ba \notag
\lambda_n &< 1- \sum_{x^n} p(x^n) P_{s_0}\left(B_0\cap D_{s_0}^c|x^n\right)\\ \notag
&= 1- \sum_{x^n} p(x^n) (P_{s_0}\left(B_0|x^n\right) -  P_{s_0}\left(B_0\cap D_{s_0}|x^n\right))\\
\label{eq.lambda_n bound}
&= P_{s_0}\left(B_{s_0}^c(X^n)\right) + \sum_{x^n} p(x^n) P_{s_0}\left(B_0\cap D_{s_0}|x^n\right)
\ea
where $B_0 = B_{s_0}(x^n)$, $D_s^c$ denotes the complement of $D_s$. Note that the 1st term in \eqref{eq.lambda_n bound} is
\ba
\label{eq.Bs^c}
t_1= P_{s_0}\left(B_{s_0}^c(X^n)\right) = \Pr\bLF i(X^n;Y^n|{s_0})<\ln\alpha \bRF
\ea
and 2nd term $t_2$ can be upper bounded as follows:
\ba \notag
t_2&= \sum_{x^n} p(x^n) P_{s_0}\left(B_0\cap D_{s_0}|x^n\right)\\ \notag &\le \sum_{x^n} p(x^n) P_{s_0}\left(D_{s_0}|x^n\right) \\ \notag
&= \sum_{x^n} p(x^n) \sum_{k=1}^{K} P_{s_0}\left(D_{k{s_0}}|x^n\right)\\ \notag
&=\sum_{k=1}^{K} \Pr\left(Y^n \in D_{k{s_0}}\right)\\ \notag
&\le \sum_{k=1}^{K} \Pr\left(Y^n \in B_{s_0}(\bu_k)\right)\\
\label{eq.K/a}
&\le K/\alpha
\ea
where we have used the facts that (i) the sets $\{D_{ks}\}_{k=1}^{K}$ are non-overlapping and (ii) $D_{ks} \in B_s(\bu_k)$. The last inequality follows from $\Pr\left(Y^n \in B_s(\bu_k)\right) \le 1/\alpha$, which follows from \eqref{eq.Pr.Bs<1/a}. Combining \eqref{eq.Bs^c} with \eqref{eq.K/a} and using \eqref{eq.lambda_n}, one finally obtains:
\ba
\lambda_n &< \Pr\bLF i(X^n;Y^n|s_0)\le \ln\alpha \bRF + K/\alpha\\ \notag
\lambda_n &= \sup_{s\in \sS} \Pr\bLF i(X^n;Y^n|s) \le \ln \alpha \bRF +
M_n/\alpha \\
&\ge \Pr\bLF i(X^n;Y^n|s_0) \le \ln \alpha \bRF +
M_n/\alpha
\ea
from which it follows that $M_n < K$.
\end{proof}

Thus, one can always select $M_n$ codewords using this iterative method. For this codebook, the maximum error probability $\varepsilon_{n,max}$ satisfies
\ba \notag
\varepsilon_{n,max} &= \sup_s \max_k P_s(D_{ks}^c|\bu_k)\\ \notag
&= \max_k \sup_s P_s(D_{ks}^c|\bu_k)\\ \notag
&= \max_k (1-\inf_s P_s(D_{ks}|\bu_k))\\
&\le \lambda_n
\ea
where $P_s(D_{ks}^c|\bu_k)$ represents error probability when $\bu_k$ is transmitted under channel state $s$ and where $\inf_s P_s(D_{ks}|\bu_k) \ge 1-\lambda_n$ by code construction. Since $\varepsilon_{n,max} \le \lambda_n$, so is the average error probability $\varepsilon_n \le \lambda_n$, from which \eqref{eq.comp.Feinstein} follows.

\subsection{Proof of Proposition \ref{prop.uuline(I)=underline(I)}}

We begin with the following Lemma.
\begin{lemma}
\label{lemma.lim sup = sup lim}
Let the sequence $f_n(s)\ge 0$ be such that $f_n(s) \rightarrow 0$ as $n\rightarrow\infty$ for any $s$. Then, the following holds if and only if the convergence is uniform,
\ba
\label{eq.lim sup = sup lim}
\lim_{n\rightarrow\infty} \sup_s f_n(s) = \sup_s \lim_{n\rightarrow\infty} f_n(s) = 0
\ea
\end{lemma}
\begin{proof}
First, note that $f_n(s) \rightarrow 0$ as $n\rightarrow\infty$ for any $s$ implies 2nd equality in \eqref{eq.lim sup = sup lim}. To prove the sufficiency for the 1st one, note that, from uniform convergence, there exists $n_0(\epsilon)$ such that
\ba
0 \le f_n(s) < \epsilon
\ea
for any $\epsilon>0$ and any $n \ge n_0(\epsilon)$. Taking $\lim_{n\rightarrow\infty}\sup_s$ of both sides, one obtains 1st equality. To prove the "only if" part, observe that the 1st equality in \eqref{eq.lim sup = sup lim} implies that for any $\epsilon >0$ there exists $n_0(\epsilon)$ such that
\bal
0 \le \sup_s f_n(s) < \epsilon\ \forall n> n_0(\epsilon)
\eal
which implies $0\le f_n(s) < \epsilon$ and hence the uniform convergence.
\end{proof}

We now show that \eqref{eq.prop.uuline(I)=underline(I)} holds for uniform compound channels. Indeed, set $R=\underline{I}(\bX_{\delta},\bY_{\delta})-\gamma$, $\gamma >0$,
\ba
f_n(s)=\Pr\bLF \frac{1}{n} i(X^n_{\delta};Y^n_{\delta}|s) \le R \bRF,
\ea and observe that
\ba
\lim_{n\rightarrow\infty}\sup_s f_n(s) = \sup_s \lim_{n\rightarrow\infty} f_n(s) = 0 \ \forall \gamma >0,
\ea
where the 1st equality is from Lemma \ref{lemma.lim sup = sup lim} and the 2nd one - from the definition of $\underline{I}(\bX,\bY)$. From this, it follows that $\uuline I(\bX_{\delta},\bY_{\delta}) \ge \underline{I}(\bX_{\delta},\bY_{\delta})$. Combining this with \eqref{eq.prop.uuline(I)<=underline(I)}, one obtains \eqref{eq.prop.uuline(I)=underline(I)}. To show the "only if" part, observe that
\ba\notag
0 &= \sup_s \lim_{n\rightarrow\infty} f_n(s)\\ \notag
&= \sup_s \lim_{n\rightarrow\infty} \Pr\bLF n^{-1} i(X^n_{\delta};Y^n_{\delta}|s) \le \uuline{I} - \gamma \bRF\\ \notag
&= \lim_{n\rightarrow\infty}\sup_s  \Pr\bLF n^{-1} i(X^n_{\delta};Y^n_{\delta}|s) \le \uuline{I} - \gamma \bRF\\
&= \lim_{n\rightarrow\infty}\sup_s f_n(s)
\ea
where 2nd and last equalities are due to $\uuline I(\bX_{\delta},\bY_{\delta}) = \underline{I}(\bX_{\delta},\bY_{\delta})$; 1st and 3rd equalities are due to the definitions of $\underline{I}(\bX_{\delta},\bY_{\delta})$ and $\uuline I(\bX_{\delta},\bY_{\delta})$. Evoking now Lemma \ref{lemma.lim sup = sup lim}, one obtains the "only if" part.

\subsection{Proof of Proposition \ref{prop.prop.XY}}

While \eqref{eq.prop.1} and \eqref{eq.prop.2} are intuitive, we give below rigorous proofs. \eqref{eq.prop.1} is proved by contradiction: assume that $\uuline{\bX} > \ooline{\bX}$, let $r=(\uuline{\bX}+\ooline{\bX})/2$, $\delta = (\uuline{\bX}-\ooline{\bX})/2>0$, so that
\ba
r= \uuline{\bX} -\delta = \ooline{\bX} + \delta
\ea
and hence
\ba\notag
0 &= \lim_{n\rightarrow\infty}\sup_s \Pr \bLF X_{ns} \le \uuline{\bX}-\delta \bRF\\ \notag
  &= \lim_{n\rightarrow\infty}\sup_s \Pr \bLF X_{ns} \le \ooline{\bX}+\delta \bRF\\ \notag
  &= 1 - \lim_{n\rightarrow\infty}\inf_s \Pr \bLF X_{ns} > \ooline{\bX}+\delta \bRF\\
  &\ge 1 - \lim_{n\rightarrow\infty}\sup_s \Pr \bLF X_{ns} \ge \ooline{\bX}+\delta \bRF    = 1
\ea
i.e. a contradiction, where 1st and last equalities are from the definitions of $\uuline{\bX}$ and $\ooline{\bX}$.

To prove \eqref{eq.prop.2}, notice that
\ba\notag
\uuline{(-\bX)} &= \sup \bLF x: \lim_{n\rightarrow\infty}\sup_s \Pr \bLF -X_{ns} \le x \bRF=0 \bRF\\ \notag
&= \sup \bLF x: \lim_{n\rightarrow\infty}\sup_s \Pr \bLF X_{ns} \ge -x \bRF=0 \bRF\\ \notag
&= -\inf \bLF z: \lim_{n\rightarrow\infty}\sup_s \Pr \bLF X_{ns} \ge z \bRF=0 \bRF\\
 &= -\ooline{(\bX)}
\ea
where $z=-x$.

To prove 2nd inequality in \eqref{eq.prop.3}, we show 1st that
\ba
\uuline{(\bX+\bY)}\le \uuline{\bX} + \ooline{\bY}
\ea
To this end, notice that proving this inequality is equivalent to proving that
\ba
\lim_{n\rightarrow\infty} \sup_s \Pr \bLF X_{ns}+Y_{ns} \le \alpha \bRF = 0
\ea
implies $\alpha \le \uuline{\bX}+\ooline{\bY}$, from which the desired inequality follows by taking $\sup$ of both sides. To prove this implication, observe that
\ba
\notag
0 &= \lim_{n\rightarrow\infty} \sup_s \Pr \bLF X_{ns}+Y_{ns} \le \alpha \bRF \\ \notag
&= \lim_{n\rightarrow\infty} \sup_s (P_{1,ns} + P_{2,ns})\\ \notag
&\ge \lim_{n\rightarrow\infty} \sup_s P_{1,ns}\\
\label{eq.P1'}
&\ge \lim_{n\rightarrow\infty} \sup_s P_{1,ns}'\\
\label{eq.P1'+P2'}
&= \lim_{n\rightarrow\infty} \sup_s (P_{1,ns}' + P_{2,ns}')\\
\label{eq.ans<=alpha-b-delta}
&= \lim_{n\rightarrow\infty} \sup_s \Pr \bLF X_{ns} \le \alpha - \ooline{\bY}-\delta \bRF =0
\ea
for any $\delta > 0$, where
\ba \notag
P_{1,ns} &= \Pr\{X_{ns}+Y_{ns} \le \alpha|B_{ns}\} \Pr\{B_{ns}\}\\ \notag
P_{2,ns} &= \Pr\{X_{ns}+Y_{ns} \le \alpha|B_{ns}^c\} \Pr\{B_{ns}^c\}\\ \notag
P_{1,ns}' &= \Pr\{X_{ns}\le \alpha-\ooline{\bY}-\delta|B_{ns}\}\Pr\{B_{ns}\}\\ \notag
P_{2,ns}' &= \Pr\{X_{ns} \le \alpha - \ooline{\bY} - \delta|B_{ns}^c\} \Pr\{B_{ns}^c\},
\ea
$B_{ns}$ denotes the event $\{Y_{ns} \le \ooline{\bY}+\delta\}$ and $B_{ns}^c$ is its complement; \eqref{eq.P1'} follows from the definition of $B_{ns}$; \eqref{eq.P1'+P2'} follows from
\ba
\lim_{n\rightarrow\infty} \sup_s P_{2,ns}' \le \lim_{n\rightarrow\infty} \sup_s \Pr\{B_{ns}^c\} = 0
\ea
where the equality follows from the definitions of $\ooline{\bY}$ and $B_{ns}^c = \{Y_{ns} > \ooline{\bY}+\delta\}$. Finally, \eqref{eq.ans<=alpha-b-delta} implies that $\alpha - \ooline{\bY}-\delta \le \uuline{\bX}$ so that $\alpha \le \uuline{\bX} + \ooline{\bY}+\delta$ for any $\delta >0$ from which $\alpha \le \uuline{\bX} + \ooline{\bY}$ follows. 2nd inequality in \eqref{eq.prop.3} follows from the symmetry of $\uuline{(\bX+\bY)}$ while the 1st inequality follows from the 2nd by observing that
\ba
\uuline{(\bX+\bY)} + \uuline{(-\bY)} = \uuline{(\bX+\bY)} - \ooline{\bY} \le  \uuline{\bX}
\ea
and re-labeling the sequences.

\eqref{eq.prop.4} follows from \eqref{eq.prop.3} via \eqref{eq.prop.2}.

\subsection{Proof of Proposition \ref{prop.zns}}

The proof consists of two parts.

\textit{Part 1}: $\uuline{\bZ} \le \tilde{Z}$. This is proved by contradiction. Assume that $\uuline{\bZ} > \tilde{Z}$ which is equivalent to $\uuline{\bZ} \ge \tilde{Z} +3\delta$ for some $\delta>0$. From the definition of $\tilde{Z}$, there are infinitely many $n$ such that $\inf_s E\{Z_{ns}\} \le \tilde{Z} + \delta/2$ and from the definition of $\inf_s$, there are such channel states $s=s(n)$ that
\ba
E\{Z_{ns(n)}\} \le \inf_s E\{Z_{ns}\} + \delta/2 \le \tilde{Z} + \delta
\ea
for all such $n$, which are denoted as $n_k,\ k=1...\infty$. Let $Z_k =  Z_{n_k s(n_k)}$ and $\tilde{Z}_{k} =  E\{Z_k\}$, and observe that
\ba
\label{eq.znks.1}
0 &= \lim_{k\rightarrow\infty}\sup_ s \Pr\{Z_{n_k s} > E\{Z_{n_k s}\} +\delta\}\\ \notag
&\ge \lim_{k\rightarrow\infty} \Pr\{Z_{k} > \tilde{Z}_{k} +\delta \}\\
\label{eq.znks.3}
&\ge \lim_{k\rightarrow\infty} \Pr\{Z_{k} > \tilde{Z} + 2\delta\} = 0
\ea
where the last equality follows from the 1st one, so that
\ba
\label{eq.znks.4}
\lim_{k\rightarrow\infty} \Pr\{Z_{k} \le \tilde{Z} + 2\delta\} = 1
\ea
where \eqref{eq.znks.1} follows from Lemma \ref{lemma.Comp.Chebyshev} below, \eqref{eq.znks.3} follows from $\tilde{Z}_{k} \le \tilde{Z} + \delta$. On the other hand,
\ba
\lim_{k\rightarrow\infty} \Pr\{Z_k \le \tilde{Z} + 2\delta\} &\le \lim_{k\rightarrow\infty} \Pr\{Z_k \le \uuline{\bZ} - \delta\} \\ \notag
&\le \lim_{k\rightarrow\infty}\sup_ s \Pr\{Z_{n_k s} \le \uuline{\bZ} -\delta \} = 0
\ea
where 1st inequality is due to $\uuline{\bZ} \ge \tilde{Z} +3\delta$, which is a contradiction to \eqref{eq.znks.4}.

\begin{lemma}[Convergence in Probability for a Compound Sequence]
\label{lemma.Comp.Chebyshev}
Let $\{Z_{ns}\}_{n=1}^{\infty}$ be a compound sequence of random variables of variance $\sigma_{ns}^2$ each such that \eqref{eq.sigmans=0} holds. Then,
\ba
\lim_{n\rightarrow\infty}\sup_ s \Pr\{|Z_{n s} - E\{Z_{ns}\}| > \varepsilon \} = 0 \ \forall \varepsilon>0
\ea
\begin{proof}
From Chebyshev inequality,
\ba
\Pr\{|Z_{n s} - E\{Z_{ns}\}| > \varepsilon \} \le \sigma_{ns}^2/\varepsilon^2
\ea
Using $\lim_{n\rightarrow\infty}\sup_ s$ on both sides results in desired equality.
\end{proof}
\end{lemma}

\textit{Part 2}: $\uuline{\bZ} \ge \tilde{Z}$. This follows from the following chain of inequalities:
\ba
0 &= \lim_{n\rightarrow\infty}\sup_ s \Pr\{Z_{n s} \le E\{Z_{ns}\} -\delta \}\\ \notag
&\ge \lim_{n\rightarrow\infty}\sup_ s  \Pr\{Z_{n s} \le \inf_s E\{Z_{ns}\} -\delta \} \\ \notag
&\ge \lim_{n\rightarrow\infty}\sup_ s  \Pr\{Z_{n s} \le \tilde{Z} - 2\delta \} = 0
\ea
for any $\delta>0$, i.e. $\uuline{\bZ} \ge \tilde{Z} - 2\delta$, which implies  $\uuline{\bZ} \ge \tilde{Z}$, where 1st equality follows from Lemma \ref{lemma.Comp.Chebyshev} and the last inequality is due to $\inf_s E\{Z_{ns}\} \ge \tilde{Z} - \delta$ for sufficiently large $n$ (from the definition of $\tilde{Z}$).

\subsection{Proof of Proposition \ref{prop.properties.I}}

To prove \eqref{eq.properties.1}, observe that
\ba
\notag
\lim_{n\rightarrow\infty}\sup_ s &\Pr\left\{\frac{1}{n} \ln \frac{p_{s x^n}(X^n)}{p_{s y^n}(X^n)} \le -\delta \right\}\\  \notag
  &= \lim_{n\rightarrow\infty} \sup_s \sum_{x^n: p_{s x^n}(x^n) \le p_{s y^n}(x^n) e^{-\delta n}} p_{s x^n}(x^n)\\ \notag
&\le  \lim_{n\rightarrow\infty} \sup_ s \sum_{x^n} p_{s y^n}(x^n) e^{-\delta n} \\
&=  \lim_{n\rightarrow\infty} e^{-\delta n} = 0 \ \forall \delta>0
\ea
from which \eqref{eq.properties.1} follows.

Eq. \eqref{eq.properties.2} follows by observing that $\uuline{I}(\bX;\bY)$ is the compound inf-divergence rate between $(\bX,\bY)$ and $(\bX',\bY')$, where $\bX'$ and $\bY'$ are independent of each other and have the same distributions as $\bX$ and $\bY$.

Eq. \eqref{eq.properties.3} follows from the symmetry of information density: $i(x^n;y^n|s)=i(y^n;x^n|s)$.

Eq. \eqref{eq.properties.4}-\eqref{eq.properties.6} follow from using $\uuline{(\cdot)}$ on
\ba
i(x^n;y^n|s) = \ln \frac{1}{p_{s}(y^n)} - \ln \frac{1}{p_{s}(y^n|x^n)}
\ea
and applying the inequalities in \eqref{eq.prop.3}. \eqref{eq.properties.4a}-\eqref{eq.properties.5a} follow from \eqref{eq.properties.4}-\eqref{eq.properties.5}.

To prove 1st inequality in \eqref{eq.properties.7}, notice that
\ba
i(x^n,y^n;z^n|s) = i(x^n;z^n|s)+i(y^n;z^n|x^n,s),
\ea
use $\uuline{(\cdot)}$ and the inequality in \eqref{eq.prop.3}. 2nd inequality follows from $\uuline{I}(\bY;\bZ|\bX) \ge 0$ and the equality part follows from
\ba
\uuline{I}(\bX,\bY;\bZ) \le \uuline{I}(\bX;\bZ) + \ooline{I}(\bY;\bZ|\bX) = \uuline{I}(\bX;\bZ)
\ea
1st inequality in \eqref{eq.properties.8} follows from $p_s(x^n|y^n) \le 1$ when the alphabet is discrete. To prove the last inequality, let $Z_{ns} = -n^{-1}\ln p_s(X^n)$ and observe the following:
\ba \notag
\Pr\{Z_{ns} \ge \ln N_x +\delta\} &= \sum_{x^n: p_s(x^n) \le e^{-n(\ln N_x +\delta)}} p_s(x^n)\\ \notag
&\le \sum_{x^n} e^{-n(\ln N_x +\delta)}\\
&= e^{-n(\ln N_x +\delta)} N_x^n = e^{-n\delta}
\ea
so that
\ba \notag
\lim_{n\rightarrow\infty} \sup_ s \Pr\{Z_{ns} \ge \ln N_x +\delta\} =0 \ea
and therefore $\uuline{H}(\bX) \le \ooline{H}(\bX) \le \ln N_x +\delta$ for any $\delta>0$, from which the desired inequality follows. This also implies the last inequalities in \eqref{eq.properties.9}-\eqref{eq.properties.10}.

2nd inequality in \eqref{eq.properties.9} follows from $\uuline{H}(\bY|\bX) \ge 0$ and \eqref{eq.properties.4}, \eqref{eq.properties.3}.

2nd inequality in \eqref{eq.properties.10} can be obtained via similar reasoning using
\ba
\ooline{I}(\bX;\bY) \le \ooline{H}(\bX) - \uuline{H}(\bX|\bY)
\ea

Eq. \eqref{eq.properties.9a} follow from \eqref{eq.properties.6}.

\subsection{Proof of Proposition \ref{prop.uulineI<=liminf.inf.I}}

Let $Z_{ns}=\frac{1}{n} i(X^n;Y^n|s)$ and observe that
\ba \notag
\frac{1}{n} &I(X^n;Y^n|s) = E\bLF Z_{ns} \bRF\\
 &\ge E\{Z_{ns}1[Z_{ns} \le 0]\} + E\{Z_{ns}1[Z_{ns} \ge \uuline{I}-\delta]\}
\ea
for any $0<\delta<\uuline{I}$, where $1[\cdot]$ is the indicator function and $\uuline{I}=\uuline{I}(\bX,\bY)$. 1st term $t_1$ can be lower bounded as follows:
\ba \notag
t_1 &= E\{Z_{ns}1[Z_{ns} \le 0]\}\\ \notag
&= \sum_{x^n,y^n: z_{ns}\le 0} p_s(y^n)p(x^n) w_{ns}\ln w_{ns}\\ \notag
&\ge -\frac{1}{n e} \sum_{x^n,y^n: z_{ns}\le 0} p_s(y^n)p_s(x^n)\\
&\ge -\frac{1}{n e}
\ea
where $w_{ns}=p_s(y^n|x^n)/p_s(y^n)$ and 1st inequality follows from $w\ln w \ge -1/e$. 2nd term $t_2$ can be lower bounded as follows:
\ba \notag
t_2 &= E\{Z_{ns}1[Z_{ns} \ge \uuline{I}-\delta]\}\\ \notag
&= \sum_{x^n,y^n: z_{ns}\ge \uuline{I}-\delta} z_{ns} p_s(y^n|x^n)p(x^n)\\ \notag
&\ge (\uuline{I}-\delta) \Pr\{Z_{ns}\ge \uuline{I}-\delta\}
\ea
Combining these two bounds, one obtains:
\ba \notag
\label{eq.Ins.I-delta}
\liminf_{n\rightarrow\infty} \inf_s \frac{1}{n} I(X^n;Y^n|s) &\ge (\uuline{I}-\delta) \lim_{n\rightarrow\infty} \inf_s \Pr\{Z_{ns}\ge \uuline{I}-\delta\} \\
&= \uuline{I}-\delta
\ea
where the equality follows from
\ba \notag
0 &=\lim_{n\rightarrow\infty} \sup_s \Pr\{Z_{ns} < \uuline{I}-\delta\} \\
&= 1 - \lim_{n\rightarrow\infty} \inf_s \Pr\{Z_{ns}\ge \uuline{I}-\delta\}
\ea
Since the inequality in \eqref{eq.Ins.I-delta} holds for each $\delta>0$, one obtains 1st inequality in \eqref{eq.uulineI<=liminf.inf.I} by taking $\delta \rightarrow 0$; 2nd one follows in the standard way.

\subsection{Proof of Proposition \ref{prop.uulineI=liminf.inf.I}}

Observe that
\ba
\label{eq.prop.uulineI=liminf.inf.I.3}\notag
E\{Z_{ns}\} &=  \overbrace{E\{Z_{ns}1[Z_{ns} \le 0]\}}^{t_1} + \overbrace{E\{Z_{ns}1[0< Z_{ns} < \uuline{I} - \delta ]\}}^{t_2}\\ \notag
 &+ \overbrace{E\{Z_{ns}1[|\uuline{I} - Z_{ns}| \le \delta]\}}^{t_3}\\ \notag
&+  \underbrace{E\{Z_{ns}1[\uuline{I} + \delta < Z_{ns} < \ln N +\delta ]\}}_{t_4}\\
 &+ \underbrace{E\{Z_{ns}1[Z_{ns} \ge \ln N +\delta ]\}}_{t_5}
\ea
where $0<\delta < \uuline{I}$, $N$ is the cardinality of either input or output alphabet (whichever is less) and $\uuline{I}=\uuline{I}(\bX,\bY)$. Let $t_1...t_5$ denote the terms on the righthand side of \eqref{eq.prop.uulineI=liminf.inf.I.3}, so that
\ba
\label{eq.prop.uulineI=liminf.inf.I.4}
\underline{\lim}\ E\{Z_{ns}\} \le \overline{\lim}\ t_1 + \overline{\lim}\ t_2 + \overline{\lim}\ t_3 + \underline{\lim}\ t_4 + \overline{\lim}\ t_5
\ea
where $\underline{\lim} = \liminf_{n\rightarrow\infty} \inf_s$ and $\overline{\lim} = \limsup_{n\rightarrow\infty} \sup_s$. It follows from the proof of Proposition \ref{prop.uulineI<=liminf.inf.I} that $t_1 \ge -1/(n e)$ so that $\overline{\lim}\ t_1 = 0$.

Without loss of generality, assume that the input alphabet is of finite cardinality and observe that the following holds:
\ba
Z_{ns} &= \frac{1}{n} \ln \frac{p_s(X^n|Y^n)}{p(X^n)} \le \frac{1}{n} \ln \frac{1}{p(X^n)}
\ea
since $p_s(x^n|y^n)\le 1$, so that
\ba \notag
E\{Z_{ns}1[Z_{ns} \ge \ln N +\delta ]\} &\le \frac{1}{n} \sum_{x^n: p(x^n)\le e^{-n\alpha}} p(x^n) \ln \frac{1}{p(x^n)} \\ \notag
&\le \sum_{x^n: p(x^n)\le e^{-n\alpha}} \alpha e^{-n\alpha}\\ \notag
&\le \alpha e^{-n\alpha} N^n\\
&= (\ln N + \delta) e^{-n\delta}
\ea
where $\alpha = \ln N +\delta$; $p(x^n)\le e^{-n\alpha}$ follows from $Z_{ns} \ge \ln N +\delta$; 2nd inequality is due to the fact that $-w \ln w$ is an increasing function if $w < 1/e$ . Taking $\lim_{n\rightarrow\infty} \sup_s$ of both sides, it follows that
\ba
\lim_{n\rightarrow\infty} \sup_s E\{Z_{ns}1[Z_{ns} \ge \ln N +\delta ]\} = 0 \ \forall \delta>0
\ea
so that $\overline{\lim}\ t_5 = 0$.

Next, observe that
\ba \notag
\label{eq.prop.uulineI=liminf.inf.I.8}
t_2 &= \sum_{x^n,y^n: 0<z_{ns}<\uuline{I} - \delta} z_{ns} p_s(y^n,x^n) \\ \notag
&\le (\uuline{I} -\delta) \sum_{x^n,y^n: 0<z_{ns}<\uuline{I} - \delta} p_s(y^n,x^n)\\
&\le (\uuline{I} -\delta) \Pr\{Z_{ns}<\uuline{I} - \delta\}
\ea
where $z_{ns}=n^{-1}i(x^n;y^n|s)$ so that
\ba
\overline{\lim}\ t_2 \le (\uuline{I} -\delta) \overline{\lim}\ \Pr\{Z_{ns}<\uuline{I} - \delta\} = 0
\ea
Using the same argument as for $t_2$, one obtains:
\ba
\underline{\lim}\ t_4 \le (\ln N +\delta) \underline{\lim}\ \Pr\{Z_{ns}> \uuline{I} + \delta\} =0
\ea
where the equality follows from \eqref{eq.prop.uulineI=liminf.inf.I.1}. Finally, one obtains:
\ba
\label{eq.prop.uulineI=liminf.inf.I.9}\notag
\underline{\lim}\ E\{Z_{ns}\} &\le \overline{\lim}\ t_3\\\notag
 &\le (\uuline{I} + \delta) \overline{\lim}\ \Pr\{|\uuline{I} - Z_{ns}| \le \delta\}\\
 & = \uuline{I} + \delta
\ea
where the equality follow from $\overline{\lim}\ \Pr\{|\uuline{I} - Z_{ns}| \le \delta\} =1$, which in turn is implied by \eqref{eq.prop.uulineI=liminf.inf.I.1}. Since \eqref{eq.prop.uulineI=liminf.inf.I.9} holds for any $\delta>0$, it follows that $\underline{\lim}\ E\{Z_{ns}\} \le\uuline{I}$, which in combination with \eqref{eq.uulineI<=liminf.inf.I} results in $\underline{\lim}\ E\{Z_{ns}\} =\uuline{I}$.

\subsection{Proof of Proposition \ref{prop.limsup.sup.I<=oolineI}}

The 1st inequality was established in \eqref{eq.uulineI<=liminf.inf.I}. The 2nd inequality is well-known. The last inequality can be established as follows. Let $i_{sn} = n^{-1} i(X^n;Y^n|s)$, $I_{sn} = E\{i_{sn}\}$, $\ooline{I}=\ooline{I}(\bX;\bY)$, $\overline{\lim} =\limsup_{n\rightarrow\infty} \sup_s$, and observe that the following chain inequality holds for any $\delta>0$:
\ba \notag
\overline{\lim}\ I_{sn} &= \overline{\lim}\ \lim_{a\rightarrow\infty} E\{i_{sn} 1[i_{sn} \le a]\} \\ \notag
&\le \limsup_{n\rightarrow\infty} \lim_{a\rightarrow\infty} \sup_s E\{i_{sn} 1[i_{sn} \le a]\} \\
\label{eq.prop.limsup.sup.I<=oolineI.4}
&= \lim_{a\rightarrow\infty} \overline{\lim}\ E\{i_{sn} 1[i_{sn} \le a]\} \\ \notag
&\le \lim_{a\rightarrow\infty} (\overline{\lim}\ E\{i_{sn} 1[i_{sn} \le \ooline{I} +\delta ]\}\\ \notag
 &\qquad + \overline{\lim}\ E\{i_{sn} 1[\ooline{I} +\delta < i_{sn} \le a]\})\\
\label{eq.prop.limsup.sup.I<=oolineI.6}\notag
&\le \lim_{a\rightarrow\infty} ((\ooline{I} +\delta) \overline{\lim}\ \Pr\{i_{sn} \le \ooline{I} +\delta\}\\ \notag
 &\qquad + a\ \overline{\lim}\ \Pr\{i_{sn}> \ooline{I} +\delta\})\\ \notag
&= \ooline{I} +\delta
\ea
where the last equality follows from $\overline{\lim}\ \Pr\{i_{sn} \le \ooline{I} +\delta\}=1,\ \overline{\lim}\ \Pr\{i_{sn}> \ooline{I} +\delta\}=0$; \eqref{eq.prop.limsup.sup.I<=oolineI.4} follows from the uniform convergence so that $\limsup_{n\rightarrow\infty} \lim_{a\rightarrow\infty} = \lim_{a\rightarrow\infty} \limsup_{n\rightarrow\infty}$; \eqref{eq.prop.limsup.sup.I<=oolineI.6} follows in the same way as in \eqref{eq.prop.uulineI=liminf.inf.I.8}. Since this chain inequality holds for any $\delta>0$, \eqref{eq.uuI<=liminfI<limsupI<=ooI} follows.

To see that the uniform convergence holds under a finite alphabet, assume, without loss of generality, that the input alphabet is finite. Then, for any $a>0$,
\ba
I_n(a) \le I_n \le I_n(a) + \Delta I_n(a)
\ea
where $\Delta I_n(a) = \sup_s E\{i_{sn} 1[i_{sn}>a]\}$, so that
\ba
|I_n - I_n(a)| \le \Delta I_n(a)
\ea
Noting that, under finite input alphabet,
\ba
i_{ns} \le Z_n = \frac{1}{n} \ln \frac{1}{p(X^n)}
\ea
one obtains for $a>\max[1, \ln N_x]$:
\ba \notag
\Delta I_n(a) &\le E\{Z_{n} 1[Z_{n}>a]\} \\ \notag
&= \frac{1}{n} \sum_{x^n: p(x^n)< e^{-na}} p(x^n) \ln \frac{1}{p(x^n)} \\ \notag
&\le \sum_{x^n: p(x^n)< e^{-na}} a e^{-na}\\ \notag
&\le a e^{-na} N_x^n =  a e^{-n(a-\ln N_x)} \\
&\le a e^{-a+ \ln N_x} \rightarrow 0
\ea
as $a\rightarrow\infty$ and the convergence is uniform in $n$ (in fact, larger $n$ imply faster convergence). 2nd inequality follows from the fact that $-w \ln w$ is an increasing function for $w<1/e$.

\subsection{Proof of Proposition \ref{prop.ineq.tildeI}}

The 1st inequality is proved by contradiction. Let $\uuline{I}= \uuline{I}(\bX;\bY),\ \check{I} = \check{I}(\bX;\bY)$, assume $\uuline{I} - \check{I} = 2 \delta >0$ and set
\ba
R = (\uuline{I} + \check{I})/2 = \uuline{I} - \delta = \check{I} + \delta
\ea
so that
\ba \notag
0 &= \lim_{n\rightarrow\infty} \sup_s \Pr\{Z_{ns} <  \uuline{I} - \delta\}\\ \notag
&= \lim_{n\rightarrow\infty} \sup_s \Pr\{Z_{ns} <R\}\\ \notag
&= 1 - \lim_{n\rightarrow\infty} \inf_s \Pr\{Z_{ns} \ge R\}\\
&= 1 - \lim_{n\rightarrow\infty} \inf_s \Pr\{Z_{ns} \ge \check{I} + \delta\} = 1
\ea
i.e. a contradiction.

The 2nd inequality is also proved by contradiction. Let $\bar{I} = \inf_s \bar{I}(\bX;\bY|s)$, assume $\check{I} - \bar{I} = 2 \delta >0$ and set
\ba
R = (\bar{I} + \check{I})/2 = \bar{I} + \delta = \check{I} - \delta
\ea
so that, from the definition of $\check{I}$,
\ba\notag
0 < \epsilon &= \limsup_{n\rightarrow\infty} \inf_s \Pr\{Z_{ns} > \check{I} - \delta\}\\ \notag
&\le \inf_s \limsup_{n\rightarrow\infty} \Pr\{Z_{ns} > \check{I} - \delta\}\\ \notag
&= \inf_s \limsup_{n\rightarrow\infty} \Pr\{Z_{ns} > \bar{I} + \delta\}\\ \notag
&\le \limsup_{n\rightarrow\infty} \Pr\{Z_{ns^*} > \bar{I} + \delta\}\\
 &\le \limsup_{n\rightarrow\infty} \Pr\{Z_{ns^*} > \bar{I}(\bX;\bY|s^*) + \delta/2\} = 0
\ea
i.e. a contradiction, where $s^*$ is such channel state that
\ba
\bar{I}(\bX;\bY|s^*) \le \inf_s \bar{I}(\bX;\bY|s) + \delta/2
\ea
The last inequality can be proved in a similar way.

To prove \eqref{eq.uulineI<=liminf.inf.I.2}, observe that
\ba \notag
\frac{1}{n} &I(X^n;Y^n|s) = E\bLF Z_{ns} \bRF\\
 &\ge E\{Z_{ns}1[Z_{ns} \le 0]\} + E\{Z_{ns}1[Z_{ns} \ge \uuline{I}-\delta]\}
\ea
for any $0<\delta<\uuline{I}$, where $1[\cdot]$ is the indicator function and $\uuline{I}=\uuline{I}(\bX,\bY)$. The 1st term $t_1$ can be lower bounded as follows:
\ba \notag
t_1 &= E\{Z_{ns}1[Z_{ns} \le 0]\}\\ \notag
&= \frac{1}{n}\sum_{x^n,y^n: z_{ns}\le 0} p_s(y^n)p(x^n) w_{ns}\ln w_{ns}\\ \notag
&\ge -\frac{1}{n e} \sum_{x^n,y^n: z_{ns}\le 0} p_s(y^n)p_s(x^n)\\
&\ge -\frac{1}{n e}
\ea
where $w_{ns}=p_s(y^n|x^n)/p_s(y^n)$ and the 1st inequality follows from $w\ln w \ge -1/e$. The 2nd term $t_2$ can be lower bounded as follows:
\ba \notag
t_2 &= E\{Z_{ns}1[Z_{ns} \ge \uuline{I}-\delta]\}\\ \notag
&= \sum_{x^n,y^n: z_{ns}\ge \uuline{I}-\delta} z_{ns} p_s(y^n|x^n)p(x^n)\\
&\ge (\uuline{I}-\delta) \Pr\{Z_{ns}\ge \uuline{I}-\delta\}
\ea
Combining these two bounds, one obtains:
\ba \notag
\label{eq.Ins.I-delta.2}
\liminf_{n\rightarrow\infty} &\inf_s \frac{1}{n} I(X^n;Y^n|s)\\ \notag
 &\ge (\uuline{I}-\delta) \lim_{n\rightarrow\infty} \inf_s \Pr\{Z_{ns}\ge \uuline{I}-\delta\} \\
&= \uuline{I}-\delta
\ea
where the equality follows from
\ba \notag
0 &=\lim_{n\rightarrow\infty} \sup_s \Pr\{Z_{ns} < \uuline{I}-\delta\} \\
&= 1 - \lim_{n\rightarrow\infty} \inf_s \Pr\{Z_{ns}\ge \uuline{I}-\delta\}
\ea
Since the inequality in \eqref{eq.Ins.I-delta.2} holds for each $\delta>0$, one obtains the 1st inequality in \eqref{eq.uulineI<=liminf.inf.I.2} by taking $\delta \rightarrow 0$. To establish the 2nd one, let $\check{I} = \check{I}(\bX;\bY)$ and observe that
\ba \notag
I_{ns}(a) = &\underbrace{E\{Z_{ns}1[Z_{ns} \le \check{I} +\delta]\}}_{e_1}\\
 &+ \underbrace{E\{Z_{ns}1[\check{I}+\delta < Z_{ns} \le a]\}}_{e_2}
\ea
for some $\delta>0$, where $1[\cdot]$ is the indicator function. The two expectation terms can be upper bounder as
\ba\notag
&e_1 \le (\check{I} +\delta)\Pr\{Z_{ns} \le \check{I} +\delta\} \\
&e_2 \le a \cdot \Pr\{Z_{ns} > \check{I} +\delta\}
\ea
so that
\ba \notag
\liminf_{n\rightarrow\infty} &\inf_s \frac{1}{n} I(X^n;Y^n|s) = \liminf_{n\rightarrow\infty} \inf_s \lim_{a\rightarrow \infty} I_{ns}(a)\\ \notag
&= \lim_{a\rightarrow \infty} \liminf_{n\rightarrow\infty} \inf_s I_{ns}(a)\\ \notag
&\le \lim_{a\rightarrow \infty} \liminf_{n\rightarrow\infty} \inf_s ( (\check{I} +\delta)\Pr\{Z_{ns} \le \check{I} +\delta\}\\ \notag
 &\mbox{\qquad} + a \cdot \Pr\{Z_{ns} > \check{I} +\delta\})\\ \notag
&\le \lim_{a\rightarrow \infty} ((\check{I} +\delta) \limsup_{n\rightarrow\infty} \sup_s  \Pr\{Z_{ns} \le \check{I} +\delta\}\\ \notag
 &\mbox{\qquad} + a \cdot \liminf_{n\rightarrow\infty} \inf_s \Pr\{Z_{ns} > \check{I} +\delta\})\\
\label{eq.strong.conv.14}
&= \check{I} +\delta
\ea
where the 2nd equality is due to uniform convergence and the last equality is due to
\ba
&\liminf_{n\rightarrow\infty} \inf_s \Pr\{Z_{ns} > \check{I} +\delta\}) = 0\\  \notag
&\limsup_{n\rightarrow\infty} \sup_s  \Pr\{Z_{ns} \le \check{I} +\delta\}\\
 &\mbox{\qquad} = 1 - \liminf_{n\rightarrow\infty} \inf_s \Pr\{Z_{ns} > \check{I} +\delta\}) = 1
\ea
Since \eqref{eq.strong.conv.14} holds for arbitrary small $\delta >0$, it follows that
\ba
\label{eq.strong.conv.14a}
\liminf_{n\rightarrow\infty} \inf_s \frac{1}{n} I(X^n;Y^n|s) \le  \check{I}
\ea
for any input.

\subsection{Proof of Theorem \ref{thm.strong.conv}}

To prove sufficiency, let the equality in \eqref{eq.strong.conv.1} to hold and select a code satisfying
\ba
\liminf_{n\rightarrow\infty} r_n =R = C_c + 3\delta
\ea
for some $\delta >0$, so that
\ba
r_n \ge R -\delta = C_c + 2\delta = \sup_{p(\bx)} \check{I}(\bX;\bY) + 2\delta
\ea
for sufficiently large $n$. Using Lemma \ref{lemma.comp.Verdu.Han} for this code, one obtains:
\ba\notag
\lim_{n\rightarrow\infty} \varepsilon_{n} &\ge \lim_{n\rightarrow\infty} \sup_{s} \Pr\bLF Z_{ns} \le r_n - \delta \bRF\\ \notag
&\ge \lim_{n\rightarrow\infty} \sup_{s} \Pr\bLF Z_{ns} \le \sup_{p(\bx)} \check{I}(\bX;\bY) + \delta \bRF\\ \notag
&\ge \lim_{n\rightarrow\infty} \sup_{s} \Pr\bLF Z_{ns} \le \check{I}(\bX;\bY) + \delta \bRF\\ \notag
&= 1- \lim_{n\rightarrow\infty} \inf_{s} \Pr\bLF Z_{ns} > \check{I}(\bX;\bY) + \delta \bRF\\
&= 1
\ea
so that \eqref{eq.strong.conv.d1} holds, where the last equality is due to
\ba
\lim_{n\rightarrow\infty} \inf_{s} \Pr\bLF Z_{ns} > \check{I}(\bX;\bY) + \delta \bRF = 0
\ea
which follows from \eqref{eq.tilde{I}}.

To prove the necessary part, assume that \eqref{eq.strong.conv.d1} holds and, using Lemma \ref{lemma.comp.Feinstein}, select a code satisfying
\ba
\lim_{n\rightarrow\infty} r_n =R = C_c + \delta
\ea
for some $\delta >0$. This implies that
\ba
r_{n} \le C_c + 2\delta
\ea
for any sufficiently large $n$. Applying Lemma \ref{lemma.comp.Feinstein}, one obtains
\ba\notag
1 = \lim_{n\rightarrow\infty} \varepsilon_{n} &\le \lim_{n\rightarrow\infty} \sup_{s} \Pr\bLF Z_{ns} \le r_{n} + \delta \bRF\\ \notag
&\le \lim_{n\rightarrow\infty} \sup_{s} \Pr\bLF Z_{ns} \le C_c + 3\delta \bRF\\
&=1
\ea
from which it follows that
\ba
\label{eq.strong.conv.10}
\lim_{n\rightarrow\infty} \inf_{s} \Pr\bLF Z_{ns} > C_c + 3\delta \bRF =0
\ea
which implies \eqref{eq.strong.conv.2a} and $\check{I}(\bX;\bY) \le C_c$ (under any input) so that, from Proposition \ref{prop.ineq.tildeI},
\ba
\label{eq.strong.conv.10a}
C_c = \sup_{p(\bx)} \uuline{I}(\bX;\bY) \le \sup_{p(\bx)} \check{I}(\bX;\bY) \le C_c
\ea
from which \eqref{eq.strong.conv.1} follows.

To establish the sufficiency of \eqref{eq.strong.conv.2a}, observe that it implies the 2nd inequality in \eqref{eq.strong.conv.10a} from which \eqref{eq.strong.conv.1} follows, which is sufficient.

To establish \eqref{eq.strong.conv.2}, observe that $C_c = \sup_{p(\bx)} \uuline{I}(\bX;\bY)$ implies that there exists such input $\bX^*$ that $\uuline{I}(\bX^*;\bY^*) > C_c -2\delta$ so that, for any such $\bX^*$,
\ba \notag
\label{eq.strong.conv.11}
0 &= \lim_{n\rightarrow\infty} \sup_s \Pr\bLF \frac{1}{n} i(X^{n*};Y^{n*}|s) < \uuline{I}(\bX^*;\bY^*) -\delta \bRF \\
&\ge \lim_{n\rightarrow\infty} \sup_s \Pr\bLF \frac{1}{n} i(X^{n*};Y^{n*}|s) < C_c -3\delta \bRF =0
\ea
Combining this with \eqref{eq.strong.conv.10} applied to input $\bX^*$, one obtains
\ba \notag
\lim_{n\rightarrow\infty} &\inf_s \Pr\{|Z_{ns}^* - C_c|> 3\delta\} \le \lim_{n\rightarrow\infty} \inf_s \Pr\{Z_{ns}^* > C_c + 3\delta\}\\
 &+ \lim_{n\rightarrow\infty} \sup_s \Pr\{Z_{ns}^* < C_c - 3\delta\} = 0
\ea
from which \eqref{eq.strong.conv.2} follows.

To establish \eqref{eq.strong.conv.1a}, apply $\sup_{p(\bx)}$ to \eqref{eq.uulineI<=liminf.inf.I.2}  to obtain
\ba \notag
C_c &= \sup_{p(\bx)} \uuline{I}(\bX;\bY)\\ \notag
 &\le \liminf_{n\rightarrow\infty} \sup_{p(x^n)} \inf_s \frac{1}{n} I(X^n;Y^n|s)\\
  &\le \sup_{p(\bx)} \check{I}(\bX;\bY) = C_c
\ea
from which the desired result follows.

\subsection{Proof of Proposition \ref{prop.uulineIe<=underlineIe}}

First, observe that
\ba \notag
\label{eq.Fx>=sup_s Fxs}\notag
\sup_s F_{\bX}(R,s) &= \sup_s \limsup_{n\rightarrow\infty} \Pr\bLF Z_{ns} \le R \bRF\\ \notag
&\le \limsup_{n\rightarrow\infty} \sup_s \Pr\bLF Z_{ns}\le R \bRF\\
 &= F_{\bX}(R)
\ea
so that
\ba
\label{eq.uulineIe<=tildeIe} \notag
\uuline{I}_{\varepsilon}(\bX;\bY) &= \sup\{R: F_{\bX}(R) \le \varepsilon\}\\ \notag
 &\le \tilde{I}_{\varepsilon}(\bX;\bY)\\
  &= \sup\{R: \sup_s F_{\bX}(R,s) \le \varepsilon\}
\ea
Next, we need the following Lemma.

\begin{lemma}
For the general compound channel, it holds that
\ba
\label{eq.tildeIe=underlineIe}
\tilde{I}_{\varepsilon}(\bX;\bY) = \underline{I}_{\varepsilon}(\bX,\bY)=\inf_s \underline{I}_{\varepsilon}(\bX,\bY|s)
\ea
\end{lemma}
\begin{proof}
Using $F_{\bX}(R,s) \le \sup_s F_{\bX}(R,s)$, observe that
\ba\notag
\Omega &= \{R: \sup_s F_{\bX}(R,s) \le \varepsilon \}\\
 &\in \Omega_{s} = \{R: F_{\bX}(R,s) \le \varepsilon \} \ \forall s
\ea
so that
\ba\notag
\tilde{I}_{\varepsilon}(\bX,\bY) &= \sup\{R: R \in \Omega \}\\
 &\le \sup\{R: R \in \Omega_s \}\\ \notag
 &= \underline{I}_{\varepsilon}(\bX,\bY|s)
\ea
and hence $\tilde{I}_{\varepsilon}(\bX;\bY) \le \underline{I}_{\varepsilon}(\bX;\bY)$. The equality is proved by contradiction. Assume that $\tilde{I}_{\varepsilon}(\bX;\bY) < \underline{I}_{\varepsilon}(\bX;\bY)$ and set $R'=(\tilde{I}_{\varepsilon}(\bX;\bY) + \underline{I}_{\varepsilon}(\bX;\bY))/2$ so that $R'> \tilde{I}_{\varepsilon}(\bX;\bY)$ and hence $\sup_s F_{\bX}(R',s)>  \varepsilon$. On the other hand,
\ba
R' <  \underline{I}_{\varepsilon}(\bX,\bY) \le \underline{I}_{\varepsilon}(\bX,\bY|s) \ \forall s
\ea
implies $F_{\bX}(R',s) \le \varepsilon \ \forall s$ so that $\sup_s F_{\bX}(R',s) \le \varepsilon$ - a contradiction.
\end{proof}

Now, combing \eqref{eq.tildeIe=underlineIe} with \eqref{eq.uulineIe<=tildeIe}, \eqref{eq.uulineIe<=underlineIe} follows. To prove the equality for an $\varepsilon$-uniform compound channel under $\bX_{\delta}$, let $Z_{ns\delta} = n^{-1} i(X^n_{\delta};Y^n_{\delta}|s)$ and establish $\uuline{I}_{\varepsilon}(\bX_{\delta};\bY_{\delta}) = \tilde{I}_{\varepsilon}(\bX_{\delta};\bY_{\delta})$:
\ba \notag
\label{eq.uulineIe=tildeIe}
\uuline{I}_{\varepsilon}(\bX_{\delta};\bY_{\delta}) &= \sup\bLF R: \limsup_{n\rightarrow\infty} \sup_s \Pr\bLF Z_{ns\delta} \le R \bRF \le \varepsilon \bRF\\ \notag
&= \sup\bLF R: \sup_s \limsup_{n\rightarrow\infty} \Pr\bLF Z_{ns\delta} \le R \bRF \le \varepsilon \bRF\\
&= \tilde{I}_{\varepsilon}(\bX_{\delta};\bY_{\delta})
\ea
where the supremum is taken over $C_{\varepsilon} -2\delta \le R \le C_{\varepsilon} +2\delta$; the 2nd equality follows from the fact that $\limsup$ and $\sup$ can be swapped for an $\varepsilon$-uniform compound channel (due to the uniform convergence property).


\end{document}